\documentclass[12pt]{article}
\usepackage{authblk}
\usepackage{geometry}
\usepackage{amsmath}
\usepackage{amssymb}
\usepackage{amsthm}
\usepackage{braket}
\usepackage{mathrsfs}
\usepackage{subfigure}
\usepackage{customcommands}
\usepackage{Nettastyle}

\title{Losing the IR: a Holographic Framework for Area Theorems
}
\author[1]{Netta Engelhardt}
\author[2]{and Sebastian Fischetti}

\affiliation[1]{Department of Physics, Princeton University, Princeton, NJ 08544, USA}
\affiliation[2]{Theoretical Physics Group, Blackett Laboratory, Imperial College, London SW7 2AZ, United Kingdom}

\emailAdd{nengelhardt@princeton.edu}
\emailAdd{s.fischetti@imperial.ac.uk}

\abstract{Gravitational area laws are expected to arise as a result of ignorance of ``UV gravitational data''.  In AdS/CFT, the UV/IR correspondence suggests that this data is dual to infrared physics in the CFT.  Motivated by these heuristic expectations, we define a precise framework for explaining bulk area laws (in any dimension) by discarding IR CFT data.  In (1+1) boundary dimensions, our prescribed mechanism shows explicitly that the boundary dual to these area laws is strong subadditivity of von Neumann entropy.  Moreover, such area laws may be of arbitrary (and mixed) signature; thus our framework gives the first entropic explanation of mixed signature area laws (as well as area laws for certain dynamical causal horizons).  In general dimension, the framework is easily modified to include bulk quantum corrections, thus giving rise to an infinite family of bulk generalized second laws.

}

\begin{document}

\maketitle

\section{Introduction}
\label{sec:intro}

The thermodynamic properties of macroscopic systems, described by effective IR theories, are typically emergent from some underlying statistical mechanical description of more fundamental UV degrees of freedom.  For this reason, the thermodynamics of gravitational systems~\cite{Haw71,Haw72, Bek72, Bek73,BarCar73, Bek74, Haw74, Bek75, Haw75, Haw76, GibHaw77a} can provide tantalizing insights into the UV completion of gravity.  In fact, our own familiar classical spacetime itself may be emergent via the same coarse-graining mechanism of UV degrees of freedom that also results in gravitational thermodynamics (see~\cite{Sei06} for a review).  A thorough understanding of this process would be an invaluable asset in the quest towards an understanding of nonperturbative quantum gravity.

Of the various thermodynamic relations of gravitating systems, the correspondence between area and entropy is the most well-understood. The Bekenstein-Hawking  entropy of a surface $\sigma$,
\be
\label{eq:BH}
S_{BH}[\sigma]=\frac{\text{Area}[\sigma]}{4G_N \hbar},
\ee
was initially studied in the context of the black hole event horizon~\cite{Bek72}; however, it has since become clear that the relationship between $S_{BH}[\sigma]$ and some coarse-graining\footnote{Here and in the rest of this paper, we caution the reader that we use the term ``coarse-graining'' in the loose sense of ``discarding information''; this notion is much more general than just discarding UV degrees of freedom to generate an effective IR theory, as the term is often used.} associated to $\sigma$ is applicable far more generally than for event horizons, in keeping with expectations about the holographic nature of gravity~\cite{Bek81, Tho93, Sus95,  StrVaf96, Bou99b, Bou99c, Bou99d, RyuTak06, HubRan07, BiaMye12, EngWal17b}.  This leads to a general expectation that area monotonicity theorems in General Relativity~\cite{Haw71,Hay93, AshKri02, BouEng15a} are manifestations of the Second Law of Gravitational Thermodynamics.  Understanding this connection precisely in the context of a particular quantum theory of gravity thus requires an appropriate notion of what it means to ``coarse-grain'' over gravitational degrees of freedom, and what constitutes an appropriate measure of the data lost in such a coarse-graining.

Our purpose in this work is to provide such a definition.  We work in the context of the AdS/CFT correspondence~\cite{Mal97, Wit98a, GubKle98}, where it is widely believed that a consistent quantum theory of gravity is defined by the boundary CFT.  By relying on certain key properties of the holographic dictionary relating the boundary and the bulk -- namely, the UV/IR correspondence~\cite{SusWit98, PeePol98}, quantum error correction~\cite{AlmDon14}, and subregion/subregion duality~\cite{Van09,Van10,CzeKar12,DonHar16,Har16,FauLew17} -- we motivate a coarse-graining framework in the boundary theory which in general gives rise to a large class of gravitational bulk area laws (of arbitrary and sometimes mixed signature).  Moreover, in~(2+1) bulk dimensions, these area laws are an immediate consequence of the strong subadditivity (SSA) of von Neumann entropy in the boundary theory: their entropic significance is manifest.

To motivate our framework, let us begin with our key question: how should we think of coarse-graining in quantum gravity?  In the context of AdS/CFT, there are two existing approaches: Wilson-like holographic RG~\cite{SusWit98, PeePol98, Akh98, BalKra99, FreGub99, DebVer99, deBo01, Lee09, FauLiu10, HeePol10} and the Jaynes maximization of entropy subject to contraints~\cite{Jay57a, Jay57b}.  The former is very well-understood, having been developed only shortly after the advent of AdS/CFT itself; moreover, it is very precise, as it can be defined purely in the (non-gravitational) boundary field theory.  However, it is primarily tasked with understanding the structure of the RG flow of the \textit{boundary} theory: given a holographic (deformed) CFT, holographic RG constructs a bulk geometrization of the RG flow of the boundary theory. The portion of this bulk geometry inside of some ``radial cutoff'' is the dual to an effective low-energy theory in the boundary.  It is immediately clear that this is precisely the \textit{opposite} of what we should want from a gravitational perspective: if gravitational area laws are to arise from coarse-graining away ``quantum gravitational degrees of freedom'', roughly speaking we must coarse-grain away the \textit{interior} of the bulk, not the asymptotic region.

On the other hand, the Jaynes-like maximization of entropy subject to constraints initially appears much more promising.  This approach defines a coarse-grained entropy of a state as
\be \label{eq:coarseentropy}
S^{\mathrm{(coarse)}}= \max\limits_{\rho \in \Hcal} S_\mathrm{vN}[\rho],
\ee
where~$\Hcal$ is a subspace of the CFT Hilbert space consisting of density matrices~$\rho$ that all satisfy some constraints, and~$S_\mathrm{vN}[\rho] = -\Tr(\rho \ln \rho)$ is the usual von Neumann entropy of~$\rho$.  Clearly, the coarse-grained entropy is expected to increase under a reduction in the number of constraints (and thus an increase in the size of~$\Hcal$). This observation was recently used by~\cite{EngWal17b} to give an entropic explanation of an area law: in AdS/CFT, the area law for spacelike holographic screens~\cite{Bou99d, BouEng15a, BouEng15b} is a thermodynamic second law of a coarse-grained entropy as defined in~\eqref{eq:coarseentropy}, with $\Hcal$ the subspace of states on which the exterior of a surface is fixed to some specified geometry (but the interior is arbitrary up to constraints)\footnote{In~\cite{KelWal13} it was suggested that a Jaynes-type coarse-graining should also compute so-called causal holographic information, which corresponds to the area of causal horizons and thus would give a coarse-graining intepretation to the Hawking area law.  Unfortunately this conjecture was falsified in~\cite{EngWal17a}.}.  However, this approach raises a philosophical dilemma: if classical spacetime is not fundamental to a quantum theory of gravity, then why should we expect that general area laws should arise from coarse-graining over spacetime \textit{regions} (such as the interiors of holographic screens)?  Indeed, we should expect that more fundamentally, we must coarse-grain over \textit{information}\footnote{A step towards this was made in~\cite{EngWal17b}, where it was conjectured that a boundary-defined quantity -- the ``simple entropy'' -- is dual to the bulk entropy coarse-grained over regions.}.

We are thus presented with a puzzle: on the one hand, the Wilsonian RG notion of coarse-graining can be phrased in terms of fundamental degrees of freedom, but coarse-grains over the wrong kind of bulk data (i.e.~the near-boundary rather than the deep bulk).  On the other hand, the Jaynes prescription for spacelike holographic screens coarse-grains over the bulk interior as desired, but relies on specifying spacetime regions, which cannot in general correspond to precise quantum gravitational degrees of freedom.  The framework that we present here is constructed by drawing only the best features from each of these two approaches: like the Wilsonian RG approach, we will phrase the framework in terms of fundamental degrees of freedom, but like the Jaynes approach of~\cite{EngWal17b}, we will make sure to coarse-grain over the bulk interior, not exterior.

We  now outline the key ingredients of framework, which will be discussed in detail in Section~\ref{sec:coarsegraining}.  First, since the boundary theory \textit{defines} the bulk quantum gravity theory, the fundamental quantum gravity degrees of freedom are just the boundary field theoretic degrees of freedom.  Thus in order to phrase the coarse-graining prescription purely in terms of fundamental degrees of freedom, we formulate it entirely in the boundary theory.  Which boundary data do we want to coarse-grain over?  Here the intuition comes from the UV/IR correspondence: the removal of UV gravitational degrees of freedom should correspond to the removal of \textit{infrared} degrees of freedom of the boundary theory.    Subregion/subregion duality gives us a clue as to how to accomplish this: since the reduced state~$\rho_R$ associated to some boundary region~$R$ is dual to the entanglement wedge~$W_E[R]$ in the bulk, restricting ourselves to access only the reduced density matrices~$\rho_\lambda \equiv \rho_{R_\lambda}$ of some family of regions~$F = \{R_\lambda\}$ (to be defined precisely below) ensures that we lose all information about the ``deep bulk''.  From the boundary perspective, in coarse-graining from a full state~$\rho$ to the set of reduced density matrices~$\{\rho_\lambda\}$, we lose the IR information about long-range correlators and entanglement, as desired; see Figure~\ref{fig:division}.  In some sense this can be thought of as a highly non-standard Wilsonian RG: since the Callan-Symanzik equation relates RG scale to the separation of points in~$n$-point functions, our procedure removes knowledge of~$n$-point functions at low energy scales but keeps knowledge of arbitrarily high-energy ones.

\begin{figure}[t]
\centering
\includegraphics[page=1]{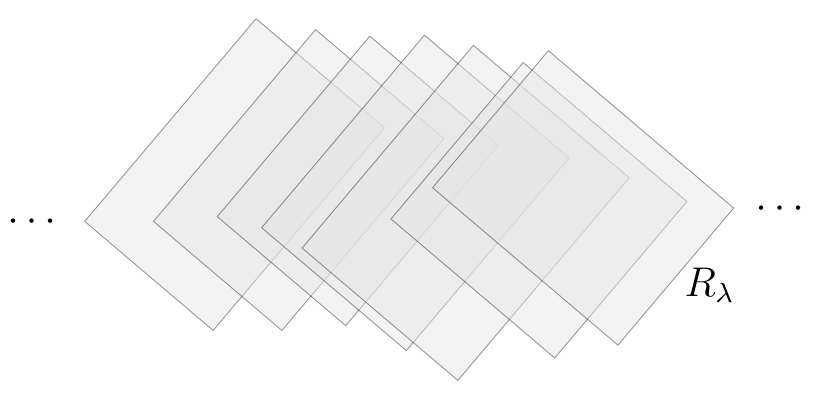}
\caption{Access to only the reduced states~$\{\rho_{R_\lambda}\}$ of some family of regions~$\{R_\lambda\}$ on the boundary removes IR data such as long-distance correlators and entanglement on the boundary.}
\label{fig:division}
\end{figure}

How do area laws arise from this prescription?  Here the understanding that AdS/CFT is a quantum error correcting code~\cite{AlmDon14} provides insight.  Coarse-graining from~$\rho$ to the~$\{\rho_\lambda\}$ introduces errors that cannot be corrected: some bulk regions become inaccessible.  If each~$\rho_\lambda$ is reduced to some~$\widetilde{\rho}_\lambda$ with~$\widetilde{R}_\lambda \subset R_\lambda$, we introduce further errors; that is, fewer messages can be decoded.  These errors are irreversible, and thus the process of continuously coarse-graining to smaller and smaller~$R_\lambda$ is a continuously irreversible process.  But such irreversible processes often result in some monotonicity property; indeed, we will see that the bulk manifestation of this monotonicity is an area theorem.  As we show in Section~\ref{sec:SSA}, in the special case of an (arbitrary) three-dimensional bulk this connection can be made completely precise.  In the boundary theory,
%YOU MAY HAVE PRECISELY ZERO REGIONS.  DO NOT PASS GO, DO NOT COLLECT TWO HUNDRED REGIONS
SSA implies that the so-called differential entropy of the family of regions~$\{R_\lambda\}$ is non-decreasing as the~$R_{\lambda}$ are shrunk.  Regardless of the physical interpretation of differential entropy (which is still lacking), its \textit{monotonicity} therefore simply \textit{is} SSA.  But by the hole-ographic prescription of~\cite{BalChoCze,BalCzeCho,MyeHea14}, differential entropy calculates the area of certain bulk surfaces constructed from the entanglement wedges~$W_E[R_\lambda]$; thus SSA (in the guise of monotonicity of differential entropy) manifests in bulk as area increase theorems.  Moreover, these area laws may be spacelike, null, or of mixed signature, yielding \textit{the first statistical entropic explanation for mixed-signature area laws and non-stationary causal horizons.} It is particularly curious that our mixed-signature area laws satisfy the same geometric properties as those of~\cite{BouEng15b}: they are constrained to flow outwards and towards the past (and the time reverse). This suggests that both area laws are a result of the same underlying mechanism.

The case of general dimension is discussed in Section~\ref{sec:generald}; though not as precise as in three bulk dimensions, we nevertheless find that for an appropriate choice of~$\{R_\lambda\}$ and~$\{\widetilde{R}_\lambda\}$ with~$\widetilde{R}_\lambda \subset R_\lambda$, there exist bulk surfaces~$\sigma$ and~$\widetilde{\sigma}$ constructed from the entanglement wedges~$W_E[R_\lambda]$,~$W_E[\widetilde{R}_\lambda]$ such that
\be
\mathrm{Area}[\widetilde{\sigma}] \geq \mathrm{Area}[\sigma].
\ee
Moreover, this area law is robust under perturbative quantum corrections: as we show in Section~\ref{sec:quantum}, in such a regime this area law becomes a Generalized Second Law (GSL)
\be
S_\mathrm{gen}[\widetilde{\sigma}] \geq S_\mathrm{gen}[\sigma],
\ee
where for the uninitiated reader, the generalized entropy~$S_\mathrm{gen}[\sigma]$ will be introduced with greater detail in Section~\ref{sec:quantum}. This natural extension under quantum corrections is strong evidence that the area laws obtained via our coarse-graining procedure are not accidental artifacts of the classical limit, but do indeed arise from some fundamental quantum gravitational features.

The present paper only scratches the surface of what can be done with our new coarse-graining framework; in Section~\ref{sec:disc} we conclude with a discussion of a number of future directions to pursue.

\textit{Note:}  The recent paper~\cite{NomRat18} may \textit{prima facie} appear to have some similarity with our results in Section~\ref{sec:generald}.  This similarity is superficial, and in fact the motivation presented here is completely opposite to that of~\cite{NomRat18}.  Here we are interested in coarse-graining away UV \textit{gravitational} degrees of freedom, which we heuristically interpret as the ``interior'' of the bulk; in~\cite{NomRat18}, on the other hand, the focus is on the more conventional coarse-graining away of UV \textit{boundary field theory} degrees of freedom, which correspond to the asymptotic region of the bulk (although a precise coarse-graining procedure is not prescribed in~\cite{NomRat18}).  Because of these complementary interpretations, the area laws obtained in~\cite{NomRat18} happen to coincide with ours in special cases, but only ours admit a precise entropic interpretation (in three bulk dimensions).

\paragraph{Preliminaries}

By a QFT, we will always mean a relativistic unitary quantum field theory.  We denote the bulk spacetime manifold by~$M$ and its boundary (on which the boundary theory lives)~$\partial M$.  We assume that the bulk has a good causal structure (e.g.~AdS hyperbolicity~\cite{Ger70}) for interpretational purposes, although our results are valid without this assumption. We use~$R$ to denote globally hyperbolic regions of~$\partial M$, which we sometimes call causal diamonds for simplicity.  The maximal development of a region~$\Sigma$ is denoted by~$D[\Sigma]$; we leave it clear from context whether this development is done in~$M$ or~$\partial M$.  Overlines~(e.g.~$\overline{R}$) will denote the complement of regions, while overlines with left subscripts will denote spatial complements: that is,~$^s \overline{R}$ denotes the set of all points spacelike-separated from all points in~$R$.  As with~$D[\Sigma]$, we will make clear from context whether these complements are taken in~$M$ or~$\partial M$.  We only assume the Null Convergence Condition to guarantee extremal wedge nesting~\cite{Wal12, EngHarTA}.  Other conventions are as in~\cite{Wald}.

The von Neumann entropy of a globally hyperbolic region~$R \subset \partial M$ in a state~$\rho$ is
\be
S_\mathrm{vN}[\rho_{\lambda}] = -\Tr\rho_R \ln\rho_R,
\ee
where~$\rho_R = \Tr_{^s \overline{R}} \rho$ is the reduced density matrix on~$R$.
By the HRT proposal~\cite{RyuTak06,HubRan07}, this entropy can be computed in a holographic state at leading order in~$1/N$ (equivalently, in~$G_N \hbar$) as
\be
S_\mathrm{vN}[\rho_R]= \frac{\text{Area}[X_R]}{4G_N\hbar},
\ee
where $X_R$ is the (bulk) minimal area extremal surface homologous to a Cauchy slice~$\Sigma_R$ of~$R_{\lambda}$. The homology constraint by definition requires the existence of a (highly nonunique) achronal hypersurface $H_R$ whose boundary is $\partial H_R = X_{\lambda} \cup \Sigma_R$.  The entanglement wedge is then defined as the AdS domain of dependence of this hypersurface:
\be
W_{E}[R] \equiv D[H_R \cap R].
\ee

\section{The Coarse-Graining Prescription}
\label{sec:coarsegraining}

Our goal is now to define the coarse-graining prescription motivated in Section~\ref{sec:intro}.  This procedure should be defined in the boundary theory, and it should behave in such a way that as we coarse-grain over more CFT data, we recover progressively less of the deep bulk interior. The holographic intuition comes from subregion/subregion duality and quantum error correction, but for purposes of generality, we define the procedure without reference to a holographic dual.

Consider, then, a QFT on a $d$-dimensional, maximally extended spacetime manifold $\partial M$.  As discussed in Section~\ref{sec:intro}, we will coarse-grain away information by introducing a continuous family of globally hyperbolic regions~$\{R_\lambda\}$ (parametrized by some set of parameters~$\lambda$) and then restricting a state~$\rho$ on~$\partial M$ to the set of reduced states~$\{\rho_\lambda\}$.  Of course, in principle we may perform such a procedure for any set of regions~$\{R_\lambda\}$, but in order to sensibly think of our procedure as coarse-graining away ``independent'' data, we would like to exclude situations like those shown in Figure~\ref{fig:bad}, where some of the~$R_\lambda$ lie in the interior or in the future of others.

\begin{figure}[t]
\centering
\includegraphics[page=2]{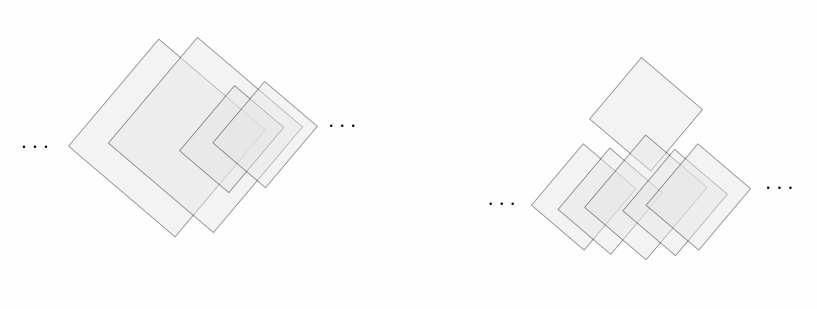}
\caption{Examples of families of regions~$\{R_\lambda\}$ that we exclude from our coarse-graining: none of the~$R_\lambda$ should be a subset of any of the others and should not lie entirely in the future or past of any of the others.}
\label{fig:bad}
\end{figure}

It is easy to require that none of the~$R_\lambda$ be contained in the future of any other; this can be accomplished by requiring that the union of the~$R_\lambda$ be bounded by two Cauchy slices~$\Sigma^\pm$ and that the boundary of each~$R_\lambda$ intersects both, as shown in Figure~\ref{fig:CauchySlice}.  More specifically, we would like to impose some notion of the~$R_\lambda$ being ``spacelike-separated'' from one another.  However, since the~$R_\lambda$ are a continuous family, they cannot be disjoint.  What, then, can we mean?  To develop some intuition, consider the case where the~$R_\lambda$ can be written as~$R_\lambda = D[I_\lambda]$, where~$I_\lambda \subset \Sigma$ are regions on some acausal slice~$\Sigma$; see Figure~\ref{subfig:regionsa}.  If none of the~$I_\lambda$ are contained in any others, the resulting family~$\{R_\lambda\}$ is one intuitively appropriate to our coarse-graining procedure.

\begin{figure}[t]
\centering
\includegraphics[page=3]{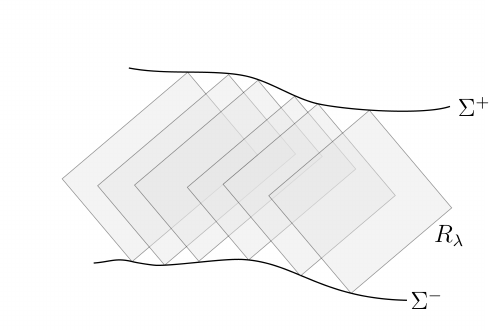}
\caption{An example of a family $\{R_{\lambda}\}$ with the two bounding Cauchy slices $\Sigma^{\pm}$.}
\label{fig:CauchySlice}
\end{figure}

\begin{figure}[t]
\centering
\subfigure[]{
\includegraphics[page=4]{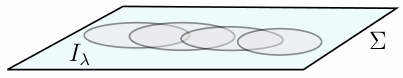}
\label{subfig:regionsa}
}
\hspace{0.5cm}
\subfigure[]{
\includegraphics[page=5]{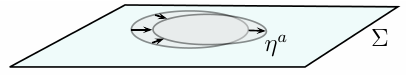}
\label{subfig:regionsb}
}
\hspace{1cm}
\subfigure[]{
\includegraphics[page=6]{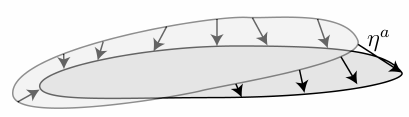}
\label{subfig:regionsc}
}
\caption{\subref{subfig:regionsa}: a family of regions suitable for our coarse-graining can be generated as the domains of dependence~$D[I_\lambda]$ of some family of regions~$I_\lambda$ contained on some acausal hypersurface~$\Sigma$.  \subref{subfig:regionsb}: such regions may be said to be spacelike-separated because any deviation vector~$\eta^a$ from one~$I_\lambda$ to an infinitesimally displaced one is everywhere spacelike.  \subref{subfig:regionsc}: even when the~$I_\lambda$ no longer all lie on~$\Sigma$, we may still think of them as ``spacelike separated'' and, thus appropriate for our coarse-graing, if all~$\eta^a$ are on average spacelike.}
\label{fig:regions}
\end{figure}

In fact, we will allow even more general families, motivated as follows (readers who are willing to accept the definition without motivation may wish to skip ahead).  First, note that since each~$R_\lambda$ is globally hyperbolic, it must admit at least one Cauchy slice~$\Sigma_\lambda$; the boundary~$\partial \Sigma_\lambda$ must in fact be independent of the choice of Cauchy slice, so we will use the notation~$\partial \Sigma_\lambda$ without ever explicitly invoking a choice of~$\Sigma_\lambda$\footnote{Alternatively, we may also write~$\partial \Sigma_\lambda = \partial R_\lambda \cap \partial \, ^s \overline{R}_\lambda$, which clearly does not require invoking any~$\Sigma_\lambda$.}.  Now, note that for each~$\partial \Sigma_\lambda$, an arbitrary variation of~$\lambda$ defines a deviation vector field~$\eta^a_\lambda$ normal to~$\partial \Sigma_\lambda$ (roughly speaking,~$\eta^a_\lambda$ encodes the ``infinitesimal deviation'' from~$\partial \Sigma_\lambda$ to~$\partial \Sigma_{\lambda + d\lambda}$).  In general the family~$\{R_\lambda\}$ will be~$(d-1)$-dimensional, and so we may introduce~$(d-1)$ linearly independent deviation vector fields on each~$\partial \Sigma_\lambda$; here we will simply treat one at a time\footnote{To be completely explicit, let~$\{\lambda^i\}$ be the~$d-1$ parameters parametrizing the family~$\{R_\lambda\}$.  Taking~$\partial \Sigma_\lambda$ to be spacelike, we may define the~$d-1$ deviation vector fields~$\eta^a_{\lambda^i} \equiv {P^a}_b (\partial_{\lambda^i})^b$, where~${P^a}_b$ is the projector normal to~$\partial \Sigma_\lambda$.  We schematically use~$\eta^a_\lambda$ to refer to any one such deviation vector field.}.

For the family~$\{R_\lambda\}$ constructed from the regions~$I_\lambda$ as above, we have~$\partial \Sigma_\lambda = \partial I_\lambda$, and thus the acausality of~$\Sigma$ implies the acausality of~$\eta^a_\lambda$.  In other words,~$\eta^2_\lambda \geq 0$ with equality only if~$\eta^a_\lambda$ vanishes.  If we further require that no two regions in the family~$\{R_\lambda\}$ coincide, then~$\eta^a_\lambda$ cannot be everywhere-vanishing.  Thus we conclude that~$\eta^2_\lambda$ must be strictly positive when integrated on~$\partial \Sigma_\lambda$ (with respect to the natural volume form):~$\int_{\partial \Sigma_\lambda} \eta^2_\lambda > 0$ for any~$\eta^a_\lambda$.  This implies that for sufficiently small deformations of the~$I_\lambda$ which move them away from a common hypersurface~$\Sigma$, the relation~$\int_{\partial \Sigma_\lambda} \eta^2_\lambda > 0$ still holds for any~$\eta^a_\lambda$, as shown in Figure~\ref{subfig:regionsb}.  In other words, we may still say that the deviation vectors~$\eta^a_\lambda$ are ``on average'' spacelike.  We will take this property, which can be invoked on \textit{any} family~$\{R_\lambda\}$, as the defining feature of what we mean by a continuous family of ``spacelike-separated'' regions.  We thus define:
\begin{defn}
Let~$F$ be a $(d-1)$-parameter continuous family\footnote{Strictly speaking this continuity requirement implies that a coarse-graining family can only exist if~$\partial M$ is connected.  The generalization to disconnected~$\partial M$ can be performed by introducing a coarse-graining family on each connected component of~$\partial M$.} of connected causal diamonds~$\{R_\lambda\}$ in $\partial M$ parametrized schematically by a set of parameters~$\lambda$.  Define~$\partial \Sigma_\lambda$ for each~$R_\lambda$ as above.  We will call~$F$ a \textit{coarse-graining family} if the following are true:
\begin{itemize}
	\item Each $\partial \Sigma_\lambda$ is everywhere spacelike;
	\item $\partial \cup_\lambda R_{\lambda}$ consists of two Cauchy slices~$\Sigma^\pm$ of~$\partial M$ with~$\partial R_{\lambda} \cap \Sigma^\pm \neq \varnothing$ for all $\lambda$;
	\item For any deviation vector field~$\eta^a_\lambda$ normal to~$\partial \Sigma_\lambda$,~$\int_{\partial \Sigma_\lambda} \eta^2_\lambda > 0$ (with the integral taken with respect to the natural volume element on~$\partial \Sigma_\lambda$).
\end{itemize}
\end{defn}

Next, recall the definition of the reduced density matrix:
\begin{equation}
\rho_{\lambda} \equiv \text{Tr}_{^s \overline{R}_\lambda}\rho,
\end{equation} 
where $\rho$ is the state on $\partial M$.  We would now like to restrict access to data (observables) which are computable from the state restricted to the causal diamonds in the coarse-graining family $F$; that is, we would like to discard data that cannot be recovered from any of the~$\rho_\lambda$.  To do this, we declare two states~$\rho$ and~$\widetilde{\rho}$ to be equivalent under IR coarse grainings associated to the family~$F$ (``$F$-equivalent'' for short) if~$\rho_{\lambda} = \widetilde{\rho}_{\lambda}$ for all $\lambda$.  The coarse-grained state~$\rho_F$ of~$\rho$ is defined as the equivalence class of~$\rho$ under this equivalence, or correspondingly as the set~$\{\rho_\lambda\}$ of reduced density matrices on~$F$.

Before discussing the ramifications of this definition, it is worth asking: is this in fact a coarse-graining?  That is, can two distinct states yield the same coarse-grained state~$\rho_F$?  It is easy to see that in discrete physical systems this is so: in the case of e.g.~a spin chain, we can take~$\lambda$ to index spin sites and each~$R_\lambda$ to be a collection of adjacent spins; then it is easy to find examples of states which are different on the whole spin chain but whose density matrices agree when reduced to any of the~$R_\lambda$.  Similarly, at any finite order in~$1/N$ in AdS/CFT, this procedure is also clearly a coarse-graining: for instance, at leading order in~$1/N$, we may consider two states whose dual bulk geometries on some Cauchy slice~$\Sigma$ agree near the boundary but differ deeper in the bulk.  If the~$R_\lambda$ are chosen sufficiently small, they can only sample the geometry in the asymptotic region, and thus the reduced states must agree\footnote{Note the importance of the family~$F$ only sampling some time strip of the boundary.  In the present example, if two states agree on a bulk Cauchy slice outside of some fiducial radial cutoff~$r_*$ but differ inside of~$r_*$, time evolution would eventually cause their geometries to differ all the way to the boundary once causal signals from within~$r_*$ propagate out.  The coarse-grained states are then guaranteed to agree only if the size of the time strip containing the family~$F$ is smaller than this propagation time.}.  Perturbative corrections in~$1/N$ (adding quantum bulk fields, gravitational dressing, etc.) proceed similarly.

We do not know if non-perturbative states of continuum QFTs also admit nontrivial equivalence classes.  But since we are working perturbatively in~$1/N$ anyway, it is sufficient for our purposes that our equivalence relation be nontrivial in AdS/CFT only to finite order in~$1/N$ (that is, we may replace the condition that~$\rho_\lambda = \widetilde{\rho}_\lambda$ with equality only to finite order in~$1/N$).  Thus for our present purposes, the map from~$\rho$ to~$\rho_F$ really does provide a coarse-graining operation.

This map has a clear interpretation as discarding IR data; for instance, a long-distance two-point correlation function of a local operator $\langle {\cal O}(x_{1}){\cal O}(x_{2})\rangle$ is inaccessible to $\rho_{F}$ whenever the $R_{\lambda}$ are sufficiently small that no single one of them contains both $x_{1}$ and $x_{2}$.  By comparison, for arbitrarily close $x_{1}$ and $x_{2}$, any coarse-graining family $F$ covering them will retain knowledge of $\langle {\cal O}(x_{1}){\cal O}(x_{2})\rangle$.  Put differently, for any smeared-out observable ${\cal O}$, there will be some coarse-grained $\rho_{F}$ with no knowledge of ${\cal O}$; however, one-point functions of local operators can always be computed no matter what the family $F$ is.  In what some readers may see as abuse of nomenclature, we shall thus refer to the map from~$\rho$ to~$\rho_F$ as an IR coarse-graining.

Recall that our goal is not just to define a coarse-graining, but also to compare coarser and finer data sets. To do this, we need to be able to compare coarse-graining families:

\begin{defn}
Let $F= \{R_{\lambda}\} $ and $\widetilde{F} = \{\widetilde{R}_\lambda\}$ be two coarse-graining families, with $\widetilde{R}_\lambda \subset R_\lambda$ for all $\lambda$. Then $\widetilde{F}$ is IR coarser than $F$.
\end{defn}

An example of such families is shown in Figure~\ref{fig:IRcoarsening}.  In fact, ultimately we are really interested in a continuous notion of coarse-graining:

\begin{figure}[t]
\centering
\includegraphics[page=7]{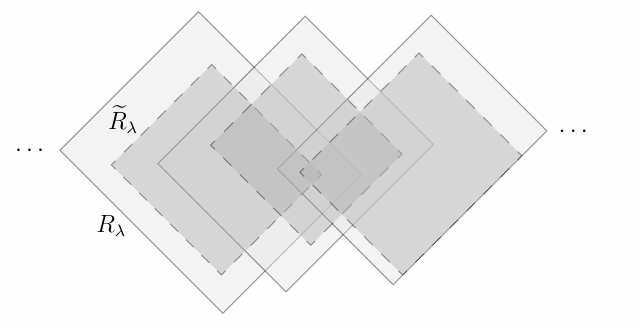}
\caption{A coarse-graining family~$F$ (light gray, solid lines) and an IR coarser family~$\widetilde{F}$ (dark gray, dashed lines).  Note that each causal diamond~$\widetilde{R}_\lambda$ lies inside (or perhaps may coincide with)~$\widetilde{R}$, so the reduced states~$\widetilde{\rho}_\lambda$ can access less data than the states~$\rho_\lambda$.}
\label{fig:IRcoarsening}
\end{figure}

\begin{defn}
Let $F= \{R_{\lambda}\} $ be a coarse-graining family, and consider continuously shrinking each $R_{\lambda}$ into some $R_{\lambda}(r)\subset R_{\lambda}$, where $r\geq 0$ indexes the continuous deformation; that is, $R_{\lambda}(r=0)=R_{\lambda}$ and $R_{\lambda}(r_{1})\subset R_{\lambda}(r_{2})$ whenever $r_{1}>r_{2}$. Let $F(r)=\{R_{\lambda}(r)\}$. Then $F(r)$ is IR coarser than $F$ for any $r>0$; we call $F(r)$ a continuous IR coarse-graining. 
\end{defn}

With the coarse-graining prescription now defined, let us analyze its bulk interpretation when the QFT is holographic.  Consider states of the CFT that describe a (semi)classical bulk dual~$(M,g_{ab})$.  Subregion/subregion duality asserts that there is an isomorphism between the algebra of operators in the entanglement wedge~$W_E[R]$ and the algebra~$\Acal[R]$ in the boundary causal diamond~$R$.  Here we will invoke subregion/subregion duality in a strong form, where we assume that all fields, \textit{including the metric}, in~$W_E[R]$ are fixed by~$\Acal[R]$ (perturbatively in~$N$, and more than a Planck length away from the boundary of $W_{E}[R]$).  This statement is known to be true at the level of quantum fields on a fixed background spacetime, as proven in~\cite{AlmDon14} and expanded upon in~\cite{Har16,FauLew17}: in this regime,~$\Acal[R]$ and the reduced density matrix $\rho_R$ can compute any observable in $W_{E}[R]$ and cannot compute any observable in $W_{E}[\, ^s \overline{R}]$\footnote{Note that this does \textit{not} include any claim about the reconstructibility of the region causally related to $W_{E}[R]$.}.  For readers skeptical of the strong version of subregion/subregion duality that we assume here, we note that our results can be restricted to work within this weaker version known to be true.

As a result of subregion/subregion duality, the bulk duals (if they exist) of two~$F$-equivalent states must agree in the region~$\cup_\lambda W_{E}[R_\lambda]$ defined by the union of the entanglement wedges of the coarse-graining family~$F$, while~$F$-equivalence implies nothing (perhaps up to constraints) about the region~$\cap_\lambda W_E[\, ^s \overline{R}_\lambda]$; see Figure~\ref{fig:wedgeunionintersections} for sketches of these regions.  Note in particular that the latter region is a ``deep bulk'' region: thus as desired, the coarse-graining from~$\rho$ to~$\rho_F$ removes data in the deep bulk.  Indeed, we may interpret this feature by borrowing intuition from quantum secret sharing properties of AdS/CFT~\cite{AlmDon14}: if we consider some bulk field $\phi(x)$ a message to be decoded and $\{R_{\lambda}\}$ as the qubits available, then having access to any of the individual $R_{\lambda}$ may not be sufficient to decode $\phi(x)$, while having access to sufficiently large unions $R_{\lambda_1} \cup \cdots \cup R_{\lambda_n}$ will indeed be sufficient. This is precisely the intuition on which we rely: as we move along a continuous IR coarse-graining $F(r)$ to larger $r$, we recover progressively less of the bulk. We may think of a continuous IR coarse-graining as a quantum error correcting code which is becoming progressively weaker\footnote{We thank A. Almheiri for discussions on this point.}.

\begin{figure}[t]
\centering
\includegraphics[page=8]{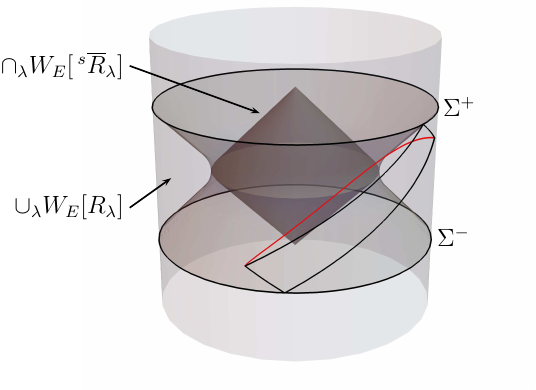}
\caption{Here we show a sketch of the bulk regions~$\cup_\lambda W_{E}[R_\lambda]$ and~$\cap_\lambda W_E[\, ^s \overline{R}]$, which are respectively completely specified and unconstrained by a coarse-grained state~$\rho_F$.  We also show one of the causal diamonds in the coarse-graining family~$F$ along with its corresponding extremal surface~$X_R$.}
\label{fig:wedgeunionintersections}
\end{figure}

\section{Implementation in $D=3$}
\label{sec:SSA}

Let us now illustrate the bulk implementation and interpretation of the coarse-graining prescription described in Section~\ref{sec:coarsegraining}.  In this section we will focus on the case of a two-dimensional boundary theory with a three-dimensional bulk dual, where the connection to monotonicity properties and area laws can be made most precise and explicit.  In fact, the precise results that we present here also hold in higher-dimensional setups with sufficient symmetry to be essentially three-dimensional; the generalization to generic higher dimensional spacetimes will be presented in Section~\ref{sec:generald}.

\subsection{Monotonicity from Strong Subadditivity}

The coarse-graining prescription presented in Section~\ref{sec:coarsegraining} was designed to discard IR CFT data: indeed, it rendered long-range correlators inaccessible and removed long-range entanglement.  Typically, under coarse-graining operations it is often useful to identify a number that roughly measures ``how much'' information is being made inaccessible.  For example, in going from a full state~$\rho$ to a particular reduced state~$\rho_R$ associated to a region~$R$, the entanglement entropy provides such a measure of information ``loss''.  Indeed, entropic inequalities often play a role in quantifying coarse-graining: for instance, SSA of von Neumann entropy, which states that the von Neumann entropies of any regions~$A$,~$B$, and~$C$ must obey
\be
\label{eq:SSA}
S(ABC) + S(B) \leq S(AB) + S(BC),
\ee
is interpreted as a statement on the irreversivility of removing subsystems. This is easily seen by rewriting~\eqref{eq:SSA} in terms of the quantum mutual information $I(A:B) \equiv S(A) + S(B)-S(AB)$, which measures correlations between states on $A$ and $B$.  SSA can be rewritten as $I(A:BC)\geq I(A:B)$: correlations are destroyed irreversibly when we discard a subsystem. 

In our case, we wish to find an object constructed from a state~$\rho$ and a coarse-graining family~$F = \{R_\lambda\}$ which can be interpreted as the amount of information lost in coarse-graining from~$\rho$ to~$\rho_F$.  Fortunately, for~$(1+1)$-dimensional field theories, a candidate already exists: the \textit{differential entropy}~\cite{BalChoCze}, which we define first for a discretized family of regions as
\begin{defn}
In any~$(1+1)$-dimensional QFT on a spacetime with compact spatial slices, let~$\{R_i\}$ with~$i = 1, \ldots, n$ be a discrete family of causal diamonds such that~$R_i \cap R_{i+1} \neq \varnothing$.  Then in any state~$\rho$, the discrete differential entropy of~$\{R_i\}$ is
\be
\label{eq:discreteSdiff}
S_\mathrm{diff}^{(n)}[\{R_i\}] = \sum_{i = 1}^n \left[S(R_i) - S(R_i \cap R_{i+1})\right],
\ee
where it is understood that~$R_{n+1} = R_1$.
\end{defn}
An illustration of the regions used to construct the discrete differential entropy is shown in Figure~\ref{fig:discreteSdiff}.  Also note that although differential entropy is computed from entanglement entropy, which is UV-divergent, these singularities cancel out in the differences in~\eqref{eq:discreteSdiff}, so the differential entropy is in fact UV-finite.

\begin{figure}[t]
\centering
\includegraphics[page=9]{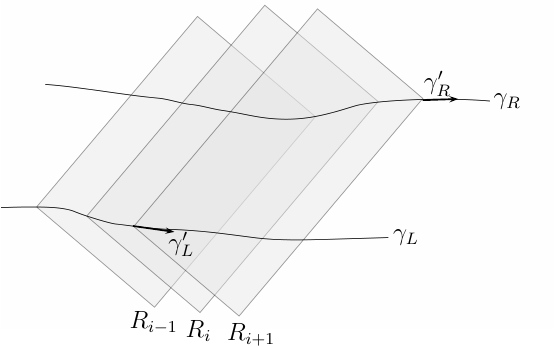}
\caption{A sketch of the discrete family of regions~$\{R_i\}$ that defines the discretized differential entropy.  In the case that~$\{R_i\}$ is a discretized coarse-graining family, the~$R_i$ are defined from the curves~$\gamma_L$,~$\gamma_R$, which must always be spacelike and such that~$\gamma'_L$ and~$\gamma'_R$ point into and out of~$R_\lambda$, respectively.}
\label{fig:discreteSdiff}
\end{figure}

In~\cite{BalChoCze} it was suggested via holographic arguments that~$S_\mathrm{diff}^{(n)}$ is a measure of the ignorance of a family of observers confined to make measurements only in the causal diamonds~$R_i$.  More precisely,~\cite{CzeHay14} showed that, at least in certain contexts, differential entropy can be interpreted in an information theoretic sense as the optimal cost of sending a state between two observers under a constrained merging protocol, with the constraint that the observers involved may only act on one of the~$R_\alpha$ at a time.  Independently of the general applicability of these interpretations, however, we claim that when the~$R_i$ are obtained from a coarse-graining family, SSA ensures that the differential entropy obeys a monotonicity property under progressive IR coarsening.  To establish this result, let us first set up some convenient notation.  In~$(1+1)$ dimensions, a coarse-graining family~$F = \{R_\lambda\}$ is parametrized by a single parameter~$\lambda$.   We will say that~$F^{(n)} = \{R_i\}$ for~$i = 1, \ldots, n$ is a \textit{discretized version} of~$F$ if for all~$i$,~$R_i = R_{\lambda_i}$ for some~$\lambda_i$ with~$\lambda_i < \lambda_{i+1}$.  Next, note that each region~$R_\lambda$ is a causal diamond which can be defined just by the positions of its left and right endpoints~$\gamma_L(\lambda)$ and~$\gamma_R(\lambda)$, or equivalently by its past and future endpoints~$\gamma^\pm(\lambda)$; as~$\lambda$ is varied, these trace out the curves~$\gamma_L$,~$\gamma_R$,~$\gamma^+$, and~$\gamma^-$.  These curves obey some useful properties:
\begin{prop}
Let~$F = \{R_\lambda\}$ be a coarse-graining family in a~$(1+1)$-dimensional spacetime, and let the curves~$\gamma_L$,~$\gamma_R$ be defined as above.  These curves are nowhere timelike, with one of the tangent vectors~$\gamma_L'$,~$\gamma_R'$ (with~$' = \partial/\partial\lambda$) always pointing towards~$R_\lambda$ and the other always away from~$R_\lambda$.
\end{prop}
\begin{proof}
In addition to~$\gamma_L$ and~$\gamma_R$, also define the curves~$\gamma_\pm$.  In $(1+1)$ dimensions, the boundary~$\partial \Sigma_\lambda$ in the definition of coarse-graining families is just the two points~$\gamma_L(\lambda)$ and~$\gamma_R(\lambda)$ and the deviation vector field~$\eta^a$ consists of the union of~$\gamma'_L(\lambda)$ and~$\gamma'_R(\lambda)$.  Thus the requirement that~$\int_{\partial \Sigma_\lambda} \eta^2 > 0$ implies that at least one of~$\gamma'_L$,~$\gamma'_R$ must be spacelike.  But if one is spacelike and the other is timelike, then at least one of~$\gamma_\pm'$ must be timelike, violating the requirement that~$\gamma_\pm$ be Cauchy slices.  Thus neither can be timelike.  Moreover, if they ever both point into or out of~$R_\lambda$, it also must be the case that at least one of~$\gamma_\pm'$ is timelike.  Thus one must always be pointing towards~$R_\lambda$ and the other always pointing away.
\end{proof}
This result allows us to unambiguously differentiate ``left'' from ``right'': we will choose left and right so that~$\gamma_L'$ points into~$R_\lambda$ and~$\gamma_R'$ points out.  Moreover, it also ensures that any discretized version~$F^{(n)}$ of~$F$ looks as shown in Figure~\ref{fig:discreteSdiff}.  It is this geometric restriction that allows us to obtain the desired monotonicity property:

\begin{thm}
\label{thm:Sdiffmonotonic}
Let~$F^{(n)}$ and~$\widetilde{F}^{(n)}$ for~$i = 1, \ldots, n$ be discretized coarse-graining families such that (i)~$\widetilde{F}$ is IR coarser than~$F$ and (ii) the unions~$R_i \cap R_{i+1}$ and~$\widetilde{R}_i \cap \widetilde{R}_{i+1}$ are non-empty.  Then in any state of any QFT,
\be
\label{eq:Sdiffmonotonicity}
S_\mathrm{diff}^{(n)}[\widetilde{F}^{(n)}] \geq S_\mathrm{diff}^{(n)}[F^{(n)}].
\ee
\end{thm}

\begin{proof}
We may think of obtaining the~$\widetilde{R}_i$ by ``pulling in'' the endpoints of each of the~$R_i$.  To show that~$S_\mathrm{diff}$ is non-decreasing under this ``pulling in'' process, it is sufficient to show that~$S_\mathrm{diff}$ is non-decreasing when only one endpoint is pulled in a small amount; symmetry guarantees that~$S_\mathrm{diff}$ is non-decreasing if the other endpoint is pulled in as well.  Consider therefore the family~$\widehat{F}^{(n)} = \{\widehat{R}_i\}$ obtained from the~$R_i$ by keeping the left endpoint unchanged but pulling the right endpoint in, as shown in Figure~\ref{fig:Sdiffproof}.

\begin{figure}[t]
\centering
\includegraphics[page=10]{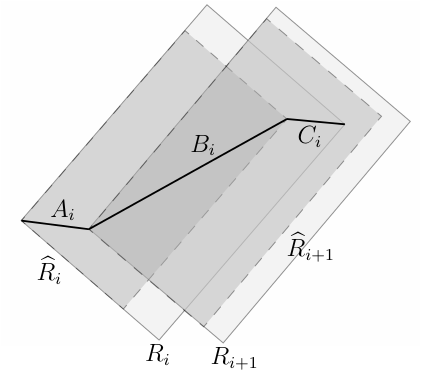}
\caption{The families of regions~$\{R_i\}$,~$\{\widehat{R}_i\}$ considered in the proof of Theorem~\ref{thm:Sdiffmonotonic}, along with the regions~$A_i$,~$B_i$,~$C_i$ to which strong subadditivity is applied.  Here each~$\widehat{R}_i$ is obtained from~$R_i$ by keeping the left endpoint fixed but pulling in the right endpoint a small amount.}
\label{fig:Sdiffproof}
\end{figure}

Now for each~$i$, we apply strong subadditivity to the regions~$A_i$,~$B_i$, and~$C_i$ defined as Cauchy slices of~$R_i \setminus R_{i+1}$,~$\widehat{R}_i \cap \widehat{R}_{i+1}$, and~$R_i \setminus \widehat{R}_i$, respectively\footnote{Note our abuse of notation: if~$A$ and~$B$ are causal diamonds such that there exists a Cauchy surface~$\Sigma$ with the property that~$A = D[A \cap \Sigma]$ and~$B = D[B \cap \Sigma]$, then here we use~$A \setminus B$ to refer to the causal diamond~$D[(A \setminus B) \cap \Sigma]$.}.  In terms of these regions, we have that~$R_i = D[A_i B_i C_i]$,~$R_i \cap R_{i+1} = D[B_i C_i]$,~$\widehat{R}_i = D[A_i B_i]$, and~$\widehat{R}_i \cap \widehat{R}_{i+1} = D[B_i]$.  Thus
\be
S^{(n)}_\mathrm{diff}[F^{(n)}] = \sum_{i = 1}^n \left[S(A_i B_i C_i) - S(B_i C_i)\right] \leq \sum_{i = 1}^n \left[S(A_i B_i) - S(B_i)\right] = S^{(n)}_\mathrm{diff}[\widehat{F}^{(n)}],
\ee
with the inequality following from strong subadditivity~\eqref{eq:SSA} applied to each term in the sum.  But clearly by repeating this process, we can continue to shrink the intermediate regions~$\widehat{R}_i$ to get all the way to~$\widetilde{F}^{(n)}$, with~$S^{(n)}_\mathrm{diff}$ increasing each time.  This proves the desired result.
\end{proof}

This result is quite remarkable, yet not unexpected: under IR-coarsening, we trace out over more and more subregions to the obtain the coarse-grained state~$\rho_F$.  As mentioned above, the irreversibility of removing more subregions is captured by SSA, and thus it is quite reasonable to expect that there should exist an object constructed from the entanglement entropies of the~$\rho_\lambda$ which behaves monotonically under IR-coarsening.  Indeed, SSA played a key role in the various entropic proofs of the~$c$,~$F$, and~$a$-theorems~\cite{CasHue04,CasHue06,CasHue12,CasHue15,CasTes17} -- which are statements of the irreversibility of coarse-graining under Wilsonian RG.  In fact, in the special case of a Poincar\'e invariant vacuum state, the discretized differential entropy~\eqref{eq:discreteSdiff} is just a sum over Casini-Huerta~$c$-functions~$S(R_i) - S(R_i \cap R_{i+1})$, so the monotonicity of~$S_\mathrm{diff}$ really does arise in the exact same way as the entropic~$c$-theorem (though here we consider arbitrary states).  We therefore interpret the monotonicity of differential entropy as confirmation that our coarse-graining procedure does what it was designed to do.

For the holographic analysis in the following section, it will be useful to note that the differential entropy is well-defined in the continuum limit~$n \to \infty$.  In this case, it takes a very simple form if we interpret the entropy~$S(R_\lambda)$ of a region~$R_\lambda$ as a function of its left and right endpoints:~$S(R_\lambda) = S(\gamma_L(\lambda),\gamma_R(\lambda))$.  With this interpretation, we define the (continuum) differential entropy of the family~$\{R_\lambda\}$ as~\cite{BalChoCze,MyeHea14}:
\begin{defn}
In any~$(1+1)$-dimensional QFT on a spacetime with compact spatial slices, let~$F = \{R_\lambda\}$ be a coarse-graining family with left and right endpoints~$\gamma_L(\lambda)$, $\gamma_R(\lambda)$.  Then the differential entropy of~$\{R_\lambda\}$, obtained in the~$n \to \infty$ limit of the discretized differential entropy~\eqref{eq:discreteSdiff}, is
\be
\label{eq:Sdiff}
S_\mathrm{diff}[F] \equiv \oint d\lambda \, \left.\frac{\partial S(\gamma_L(\lambda),\gamma_R(\lambda'))}{\partial \lambda'}\right|_{\lambda' = \lambda} = - \oint d\lambda \, \left.\frac{\partial S(\gamma_L(\lambda'),\gamma_R(\lambda))}{\partial \lambda'}\right|_{\lambda' = \lambda};
\ee
the two expressions are equal as can be seen via integration by parts.
\end{defn}

This continuum expression conveniently makes clear that differential entropy is not a positive quantity: for instance, since for pure states~$S[R_\lambda] = S[\, ^s \overline{R}_\lambda]$, we have~$S_\mathrm{diff}[F] = -S_\mathrm{diff}[\{\, ^s\overline{R}_\lambda\}]$.  However, it is also clear that if the regions~$R_\lambda$ are sufficiently small relative to any other scales,~$S(\gamma_L(\lambda),\gamma_R(\lambda))$ will behave similarly to how it does in the vacuum, and thus it will increase as~$\gamma_L(\lambda)$ and~$\gamma_R(\lambda)$ are moved apart; this implies that for sufficiently small~$R_\lambda$,~$S_\mathrm{diff}$ is positive.  We may interpret this result result heuristically as follows: starting with small~$R_\lambda$, our monotonicity result implies that~$S_\mathrm{diff}$ decreases as the size of the~$R_\lambda$ increases.  Eventually,~$S_\mathrm{diff}$ may vanish and become negative; this is an indication that the regions~$R_\lambda$ have become large enough that no IR data is lost in the coarse-graining to~$\rho_F$ (note that since we are assuming compact spatial slices, the volume of these slices imposes a natural IR cutoff).  Further increasing the size of the~$R_\lambda$ morally does not recover any new information.  This heuristic interpretation of negative differential entropy is pleasantly consistent with that of~\cite{CzeHay14}, in which a negative differential entropy corresponds to a distillation (rather than consumption) of entanglement in the constrained merging protocol.

\subsection{An Abundance of Area Laws}

Although Theorem~\ref{thm:Sdiffmonotonic} makes no reference to holography, a remarkable consequence of it is that in a holographic setting, the bulk dual to it is an area law!  In fact, Theorem~\ref{thm:Sdiffmonotonic} immediately gives rise to an \textit{infinite} family of area laws in the bulk, consistent with the expectation laid out in Section~\ref{sec:intro}.  This observation follows from the hole-ographic interpretation of differential entropy:~$S_\mathrm{diff}$ computes the area (or more generally, a local geometric functional) of a particular curve (or in general, curves) in the three-dimensional bulk dual.  The full details of this connection can be found in~\cite{MyeHea14}; here we offer a brief review of the salient features for convenience to the reader.

Assume the existence of an asymptotically AdS dual to the two-dimensional field theory state~$\rho$, and consider a one-parameter family~$\{\Gamma_\lambda\}$ of geodesics in the AdS spacetime such that each~$\Gamma_\lambda$ is anchored at the AdS boundary at the points~$\gamma_L(\lambda)$,~$\gamma_R(\lambda)$.  The family~$\{\Gamma_\lambda\}$ defines a deviation vector~$h^a_\lambda$ on each~$\Gamma_\lambda$; we take~$h^a_\lambda$ to be normal to~$\Gamma_\lambda$ (concretely,~$h^a_\lambda = {P^a}_b(\partial_\lambda)^b$, where~${P_a}^b$ is the orthogonal projector to~$\Gamma_\lambda$).  As discussed in detail in~\cite{MyeHea14}, for the families of boundary regions we consider (where~$\gamma'_L(\lambda)$ points into~$R_\lambda$ and~$\gamma'_R(\lambda)$ points out of~$R_\lambda$), we are guaranteed that~$h^2_\lambda$ will vanish somewhere on~$\Gamma_\lambda$, corresponding to~$h^a_\lambda$ becoming null or, non-generically, vanishing entirely.  As illustrated in Figure~\ref{fig:sigmaB}, now consider a bulk curve~$\sigma_B$ defined by taking the union of such points on~$\Gamma_\lambda$; that is, letting~$s$ be a parameter along each~$\Gamma_\lambda$, define~$s^*(\lambda)$ such that~$h^2_\lambda(s^*(\lambda)) = 0$, and then define~$\sigma_B(\lambda) = \Gamma_\lambda(s^*(\lambda))$.  Geometrically, this construction ensures that~$\sigma_B$ intersects each~$\Gamma_\lambda$, and where it does the tangent vectors to~$\sigma_B$ and to~$\Gamma_\lambda$ span a null plane.  (In general there may be more than one choice of~$s^*(\lambda)$, and thus more than one such~$\sigma_B(\lambda)$ may be defined; what follows is true for any particular choice as long as~$\sigma_B(\lambda)$ is connected.)

\begin{figure}[t]
\centering
\includegraphics[page=11]{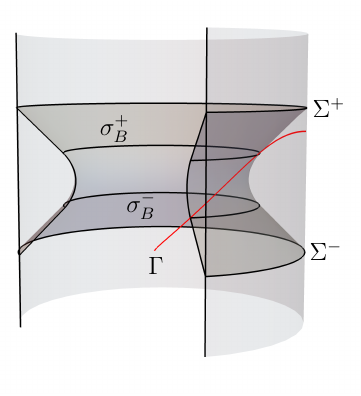}
\caption{An illustration of the construction of the bulk curves~$\sigma_B^\pm$ from the boundary regions~$R_\lambda$ and corresponding family of bulk curves~$\Gamma_\lambda$.  Here we show a single member~$\Gamma$ of the family~$\{\Gamma_\lambda\}$; the two bulk curves labeled~$\sigma_B^\pm$ have the property that where they intersect~$\Gamma$, the tangent vectors to~$\Gamma$ and to~$\sigma_B^\pm$ span a null plane.}
\label{fig:sigmaB}
\end{figure}

The main result of~\cite{MyeHea14} is then that in the regime where the lengths of the~$\Gamma_\lambda$ compute the boundary entanglement entropies~$S(\gamma_L(\lambda),\gamma_R(\lambda))$ via the HRT formula, as long as~$\sigma_B$ is differentiable and everywhere spacelike,
\be
\label{eq:holeography}
|S_\mathrm{diff}[R_\lambda]| = \frac{\Len(\sigma_B)}{4G_N \hbar}.
\ee
(The absolute value is required because as noted above~$S_\mathrm{diff}$ need not be positive.)  In other words, given a family~$\{R_\lambda\}$, we can define a bulk curve~$\sigma_B$; if~$\sigma_B$ is everywhere smooth and spacelike, its length is computed from the differential entropy of~$\{R_\lambda\}$.

More generally,~$\sigma_B$ may become null somewhere or have cusps; in such cases,~$S_\mathrm{diff}[R_\lambda]$ in fact computes a \textit{signed} length of~$\sigma_B$, where portions of different sign are joined wherever~$\sigma_B$ fails to be smooth and spacelike.  Indeed, this sign ambiguity is the need for the absolute value in~\eqref{eq:holeography}: even when~$\sigma_B$ is everywhere spacelike and differentiable,~$S_\mathrm{diff}[R_\lambda]$ could compute the \textit{negative} of its length.  Such signed lengths may be physically interesting (for instance, they contribute to the change in sign of~$S_\mathrm{diff}$, which we argued heuristically above can potentially be understood as an indication that no IR data is being lost), but for simplicity we will hereafter restrict to the case where~$\sigma_B$ is everywhere spacelike and differentiable.

We now immediately obtain an infinite class of area laws.  Consider \textit{any} coarse-graining family~$F$ such that~$S_\mathrm{diff}[F] > 0$, and introduce an IR-coarser family~$\widetilde{F}$.  If the state has a classical geometric dual and the curves~$\sigma_B$,~$\widetilde{\sigma}_B$ constructed from~$F$ and~$\widetilde{F}$ are everywhere differentiable\footnote{We will relax the requirement that~$\sigma_B$ and~$\widetilde{\sigma}_B$ be everywhere differentiable in Section~\ref{sec:generald}; here we need this requirement to develop our precise entropic interpretation.} and spacelike, then monotonicity of the differential entropy~\eqref{eq:Sdiffmonotonicity} combined with the hole-ographic formula~\eqref{eq:holeography} implies that
\be
\Len[\widetilde{\sigma}_B] \geq \Len[\sigma_B].
\ee
(On the other hand, if~$S_\mathrm{diff}[F] \le 0$, then we obtain~$\Len[\widetilde{\sigma}_B] \leq \Len[\sigma_B]$ as long as~$S_\mathrm{diff}[\widetilde{F}] \leq 0$ as well.)  Recall that the inequality is simply strong subadditivity: \textit{the removal of long-range entanglement in the boundary maps precisely to an area law in the bulk!}

In fact, this construction can be slightly generalized to higher dimensions.  Assume that instead of being geodesics, the curves~$\Gamma_\lambda$ extremize a geometric action
\be
\label{eq:Lbulk}
L_\lambda[\gamma] \equiv \int_{s_L}^{s_R} ds \, \Lcal(\gamma(s), \gamma'(s)) \mbox{ with } \gamma(s_{L,R}) = \gamma_{L,R}(\lambda),
\ee
such that~$L_\lambda[\Gamma_\lambda] = S(\gamma_L(\lambda),\gamma_R(\lambda))$,~$\Lcal(\gamma(s), \gamma'(s))$ is positive and depends only on~$\gamma$ and its tangent vector~$\gamma'$, and~$L_\lambda$ is invariant under reparametrizations of~$s$.  The more general result of~\cite{MyeHea14} is that
\be
|S_\mathrm{diff}[R_\lambda]| = \oint d\lambda \, \Lcal(\sigma_B(\lambda),\sigma_B'(\lambda)),
\ee
where~$\sigma_B$ is a bulk curve constructed from~$\{\Gamma_\lambda\}$ in the same way as above.  Now, in higher-dimensional geometries, if the family~$\{\Gamma_\lambda\}$ obeys a symmetry property dubbed generalized planar symmetry in~\cite{MyeHea14}, then the problem of computing codimension-two extremal surfaces essentially reduces to solving for curves that extremize an action of the form~\eqref{eq:Lbulk}.  In these (very restricted) higher-dimensional setups, differential entropy then computes the area of a codimension-two extremal surface in the bulk, and we again obtain an area law.  However, because the requirement of generalized planar symmetry is so strong, we will continue to restrict only to the case of two boundary dimensions for the remainder of this section.

As a final note, it is worth remarking once again that since the physical role of~$S_\mathrm{diff}[R_\lambda]$ is not well-understood, we do not purport to give a physical interpretation to the curves~$\sigma_B$ or their area.  Rather, we have derived an entropic understanding of the \textit{monotonicity} in the area of a family of curves.  To illustrate this construction, we now turn to some explicit examples that explain old and novel area laws in AdS$_3$.

\subsection{Spacelike Area Laws in Pure AdS}

We will make use of global coordinates~$(t,r,\phi)$ in terms of which the metric of pure AdS$_3$ is
\be
\label{eq:AdS3}
ds^2 = -\left(1+\frac{r^2}{\ell^2}\right)dt^2 + \frac{dr^2}{1+r^2/\ell^2} + r^2 d\phi^2.
\ee
It will sometimes be useful to convert to a compactified coordinate $r_* = \ell \arctan(r/\ell)$, in terms of which the metric becomes
\be
ds^2 = \sec^2(r_*/\ell)\left(-dt^2 + dr_*^2\right) + \ell^2 \tan^2(r_*/\ell) \, d\phi^2.
\ee
Clearly, null radial geodesics are just given by lines of constant~$t \pm r_*$.

First, consider working on a static time slice~$t = $~const.  This is just the Poincar\'e disk, on which the construction of the bulk curves~$\sigma_B$ has been studied extensively; see for instance~\cite{CzeLam,CzeLamMcC15a}.  Spacelike geodesics on this slice are given by
\be
\label{eq:staticspacelikegeodesics}
\tan^2 (\phi - \phi_0) = \frac{r^2 \tan^2 (\Delta \phi/2) - \ell^2}{r^2 + \ell^2},
\ee
where the endpoints of the geodesics lie at~$\phi = \phi_0 \pm \Delta \phi/2$ on the boundary.  Thus the minimum~$r$ reached by a geodesic whose endpoints have angular separation~$\Delta \phi$ is
\be
\label{eq:rmin}
r_\mathrm{min}(\Delta \phi) = \ell |\cot(\Delta \phi/2)|, \qquad (r_*)_\mathrm{min}(\Delta \phi) = \frac{\ell}{2}\left|\pi - \Delta \phi\right|.
\ee

Therefore, consider first a set of boundary intervals all of the same angular extent~$\Delta \phi \leq \pi$; the bulk curve~$\sigma_B$ defined by these intervals is just a circle of radius~$r_\mathrm{min}(\Delta \phi)$, and the differential entropy just computes the circumference of this circle~$2\pi r_\mathrm{min}(\Delta)$.  Now, as~$\Delta \phi$ is decreased, so that the boundary intervals all become smaller, the circumference of the corresponding bulk circle clearly increases: this gives a spacelike area law.  On the other hand, it is worth noting that for~$\Delta \phi > \pi$, the differential entropy of the intervals of size~$\Delta \phi$ computes the \textit{negative} circumference of the bulk circle; in this case, decreasing~$\Delta \phi$ initially decreases the circumference of the bulk circle, which is consistent with monotonicity of the differential entropy: the \textit{negative} circumference still increases.

More generally, given any closed differentiable bulk curve~$\sigma_B$ with no self-intersections on the Poincar\'e disk, a family of spatial boundary intervals~$I_\lambda$ whose differential entropy computes the length of~$\sigma_B$ can be found by just firing tangent geodesics off of~$\sigma_B$, as shown in Figure~\ref{fig:spacelikegeneral}.  However, only if~$\sigma_B$ is convex are the causal diamonds~$R_\lambda = D[I_\lambda]$ of the resulting boundary intervals a coarse-graining family; this means that only if~$\sigma_B$ is convex are we guaranteed by SSA that it obeys an area law.  The reason for this matches beautifully with geometric intuition: if the boundary intervals are shrunk, then~$\sigma_B$ moves towards its exterior.  Now, if~$\sigma_B$ is convex, its outwards-directed expansion is non-negative, and thus any deformation of it towards its exterior cannot decrease its area.  On the other hand, if~$\sigma_B$ is concave, it must have at least some portions where its outwards-directed expansion is negative: deforming just these regions outwards would \textit{decrease} its area, so it cannot obey a general area law.  Thus we see that the definition of a coarse-graining family automatically excludes concave curves, which would violate a potential area law.

\begin{figure}[t]
\centering
\includegraphics[page=12]{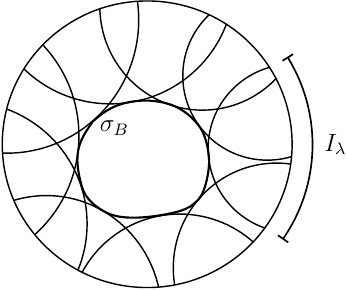}
\caption{From a general bulk curve~$\sigma_B$, we can always find a family~$\{I_\lambda\}$ of boundary intervals whose differential entropy computes the length of~$\sigma_B$.  If~$\sigma_B$ is convex, these intervals define a coarse-graining family and thus yield a monotonicity property of differential entropy; this corresponds to the fact that only if~$\sigma_B$ is convex are we guaranteed that any ``outwards'' deformation of~$\sigma_B$ will increase its length.}
\label{fig:spacelikegeneral}
\end{figure}

Finally, let us note that although here we focused on the Poincar\'e disk, we could of course consider slightly modifying all the intervals~$I_\lambda$ so that they don't all lie on the same time slice.  As long as the resulting causal diamonds~$R_\lambda$ still constitute a coarse-graining family, it is clear that generically the corresponding bulk curve~$\sigma_B$ will no longer lie on some slice of time symmetry.  Nevertheless, as long as the modifications to the~$I_\lambda$ are sufficiently small, the resulting area laws obtained by shrinking the causal diamonds will still be spacelike.  If the intervals are deformed sufficiently far from all lying on a constant time slice, the signature of the area law may change; this leads us to the next section.

\subsection{Null Area Laws in Pure AdS}

Consider again a family of intervals~$\{I_\lambda\}$ of angular size~$\Delta \phi = \pi$ on the slice~$t = 0$.  The corresponding bulk geodesics pass through the bulk point~$(t,r) = (0,0)$, and thus the bulk ``curve''~$\sigma_B$ degenerates to a point; the length of~$\sigma_B$ vanishes.  Next, consider the intervals defined by the intersection of the causal diamonds~$D[I_\lambda]$ with slices of constant time~$t > 0$; these define new spatial intervals of size~$\Delta \phi(t) = \pi - 2t/\ell$.  From~\eqref{eq:rmin}, the corresponding bulk curves~$\sigma_B$ are circles at~$r_*(t) = t$.  But since outgoing radial bulk null geodesics correspond to lines of constant~$t - r_*$, these circles correspond to constant-$t$ slices of the future lightcone of the point~$(t,r) = (0,0)$: in other words, the family of intervals of size~$\Delta \phi(t)$ generate bulk curves~$\sigma_B(t)$ which trace out a lightcone, as shown in Figure~\ref{subfig:nullcircle}.

We may consider more general slices of this light cone as follows.  This light cone is generated by radial geodesics fired from the point~$(t,r) = (0,0)$; since they are radial, these generators can be labeled by~$\phi$.  Any (spatial) slice~$\gamma$ of the light cone intersects each of these generators only once, and thus we may parametrize any such slice by the time~$t(\phi)$ at which the generator at angle~$\phi$ intersects~$\gamma$. Now consider a family of intervals parametrized by~$\phi$, with each interval centered at~$\phi$ lying on the time slice~$t(\phi)$ and with angular size~$\Delta \phi(\phi) = \pi - 2t(\phi)/\ell$.  By construction, each geodesic anchored to these intervals will intersect the light cone precisely on the slice~$\gamma$, and thus at this point, the tangent vector to~$\gamma$ and to the boundary-anchored geodesic span a null plane.  By the hole-ographic construction, this ensures that the bulk curve~$\sigma_B$ constructed from these boundary intervals will correspond precisely to~$\gamma$.  Thus differential entropy can be used to compute the area of \textit{any} slice of the light cone, not just symmetric ones; see Figure~\ref{subfig:nullgeneral}.

\begin{figure}[t]
\centering
\subfigure[]{
\includegraphics[page=13]{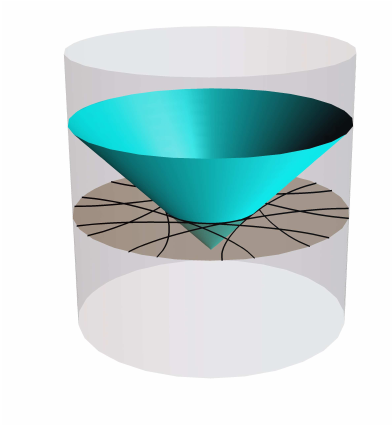}
\label{subfig:nullcircle}
}
\hspace{1cm}
\subfigure[]{
\includegraphics[page=14]{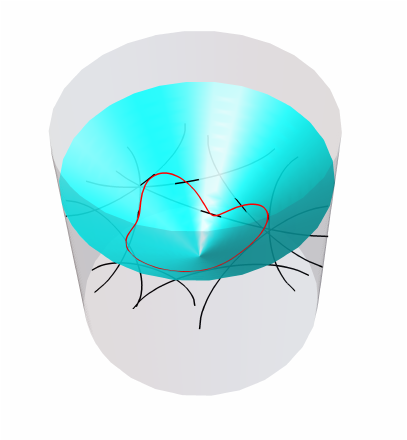}
\label{subfig:nullgeneral}
}
\caption{\subref{subfig:nullcircle}: circular cross sections of the light cone of a point in AdS can be obtained as the curves~$\sigma_B$ generated from intervals of angular extent~$\Delta \phi(t) = \pi -2t/\ell$ on the time slice at time~$t$.  \subref{subfig:nullgeneral}: by allowing the boundary intervals to lie on different time slices relative to each other, we can obtain any cross-section of the light cone.  In both cases, monotonicity of differential entropy applies, giving an entropic origin of the area law for this light cone.}
\label{fig:nullAdS}
\end{figure}

As the slice~$\gamma$ is moved to the future, the corresponding boundary intervals shrink into themselves, and thus their differential entropy must be nondecreasing.  Thus the monotonicty of the area of any slice of the light cone corresponds directly to the monotonicity of differential entropy.  Moreover, note that this light cone also happens to be a causal horizon, and therefore we have an explanation for the Hawking area law along a causal horizon.  This may be a simple case, but to our knowledge it is the first entropic understanding of the Hawking area law for non-stationary causal horizons.

\subsection{Mixed-Signature Area Laws in Pure AdS}

To obtain the spacelike and null area laws above, we modified the boundary intervals generating the bulk curves~$\sigma_B$ by changing their size and by translating them in time.  Here we explore the final degree of freedom -- boosts -- and show that we can obtain area laws for mixed-signature surfaces: that is, surfaces with timelike, null, and spacelike components.

To that end, let us introduce boundary null coordinates~$u = t/\ell + \phi$,~$v = t/\ell - \phi$, and consider boundary causal diamonds defined by left and right endpoints with null separations~$\Delta u$,~$\Delta v$.  In order to ensure that these points are spacelike separated, we require~$0 < \Delta u < 2\pi$ and~$0 > \Delta v > -2\pi$.    Consider now a family of such causal diamonds with~$\Delta u$ and~$\Delta v$ fixed, but centered on the time slice~$t = 0$, as shown in Figure~\ref{fig:boostedintervals}.  In fact, since the differential entropy of regions with~$\Delta \phi > \pi$ is just the negative of that of regions with~$\Delta \phi \to 2\pi - \Delta \phi$, we may simply restrict to~$\Delta \phi \leq \pi$, implying~$\Delta u \leq 2\pi - |\Delta v|$.  Now by keeping~$\Delta v$ fixed while varying~$\Delta u$ (or vice versa), we may stretch or contract these intervals in one null direction while leaving their extent in the other null direction unchanged: it is this deformation that we will exploit to construct a mixed-signature area law.

\begin{figure}[t]
\centering
\includegraphics[page=15]{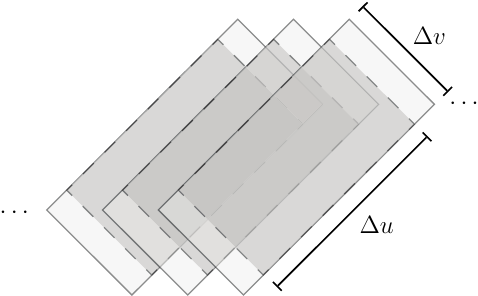}
\caption{The boosted intervals used in the construction of mixed-signature area laws.  For fixed~$\Delta v$, decreasing~$\Delta u$ shrinks the intervals in the~$u$ null direction while leaving their extent in the~$v$ null direction unchanged.}
\label{fig:boostedintervals}
\end{figure}

We leave details of the construction to Appendix~\ref{app:mixedarealaw}; the punchline is that such a family of intervals generates two curves~$\sigma_B^\pm$ lying at
\bea
t_{\sigma_B^\pm} &= \pm \ell \arctan\left|\cot\left(\frac{\Delta u - \Delta v}{4}\right)\tan^2\left(\frac{\Delta u + \Delta v}{4}\right)\right|, \label{subeq:timelikearealawt} \\
r_{\sigma_B^\pm} &= \frac{\ell}{2}\left|\cot\left(\frac{\Delta u}{2}\right) - \cot\left(\frac{\Delta v}{2}\right)\right|.\label{subeq:timelikearealawr}
\eea
When~$\Delta u = 2\pi - |\Delta v|$, it is clear from the above that~$\sigma_B^\pm$ degenerate to a point at~$r_{\sigma_B} = 0$, and thus have zero length.  It is also clear that as~$\Delta u$ is decreased (with~$\Delta v$ fixed),~$\sigma_B^\pm$ move monotonically to increasing~$r$, and therefore have monotonically increasing areas: we obtain area laws, as required by positivity of differential entropy.

What is more interesting is the signature of these area laws.  Indeed, consider the surfaces~$H^\pm$ swept out by~$\sigma_B^\pm$ as~$\Delta u$ is varied between~$0$ and~$2\pi - |\Delta v|$; these surfaces are shown in Figure~\ref{fig:timelikeAdS}.  As can be seen in the plot and as we discuss in more detail in Appendix~\ref{app:mixedarealaw},~$H^\pm$ are always spacelike at small and large~$r$; however, if~$\Delta v > -2\arcsin(\sqrt{2}/3) \approx -0.31 \pi$, then~$H^\pm$ will also be \textit{timelike} at some intermediate~$r$.  Thus for sufficiently small~$|\Delta v|$, we have constructed area laws for surfaces of \textit{mixed signature}; this is the first entropic explanation for such area laws.

\begin{figure}[t]
\centering
\subfigure[]{
\includegraphics[page=16]{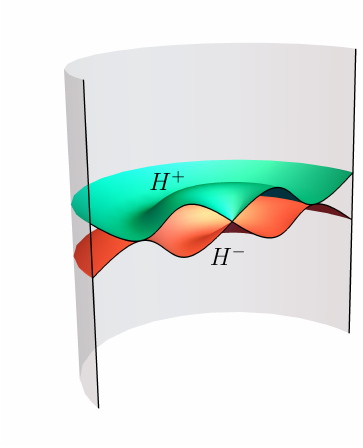}
\label{subfig:spacelikelaw}
}
\hspace{1cm}
\subfigure[]{
\includegraphics[page=17]{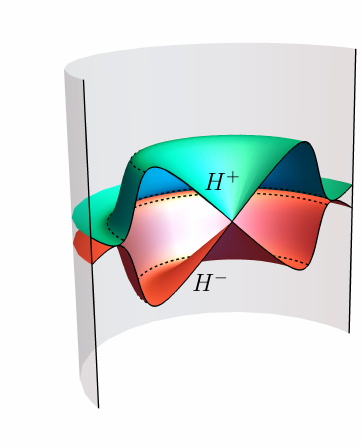}
\label{subfig:mixedlaw}
}
\caption{The surfaces~$H^\pm$ swept out by~$\sigma_B^\pm$ as~$\Delta u$ is varied with~$\Delta v$ held fixed.  In~\subref{subfig:spacelikelaw},~$\Delta v = -\pi/3$, and~$H^\pm$ are everywhere spacelike.  In~\subref{subfig:mixedlaw},~$\Delta v = -\pi/6$, and~$H^\pm$ are timelike between the surfaces shown as dotted lines.}
\label{fig:timelikeAdS}
\end{figure}

It is tempting to compare this area law to the area law for holographic screens~\cite{BouEng15a, BouEng15b}, which also have mixed signature.  In fact, future holographic screens may have spacelike and timelike components, but on the timelike components the area must increase towards the past: this behavior is reproduced by the surface~$H^+$ (the surface~$H^-$ behaves like past holographic screens, in which the area of timelike portions increases to the future).  Indeed, this behavior is general: for a general family of intervals in any geometry, shrinking the size of the intervals must result in a smaller (that is, nested) entanglement wedge.  Since the differential entropy curves~$\sigma_B$ must be tangent to the entanglement wedges of the boundary intervals that generate them, the ``futuremost'' differential entropy curve~$\sigma_B^+$ can only move in spacelike or past-timelike directions when the boundary intervals are shrunk.  (In non-generic cases like the null area law in pure AdS, it may also move in the future-directed null direction, but even there it may never move in a future-directed timelike direction.)  This connection between the behavior of area laws constructed from differential entropy and those of holographic screens is tantalizing, and is worth exploring further.

Finally, note that although we derived this area law in pure AdS, it is robust under sufficiently small perturbations of either the boundary intervals or of the bulk geometry, since sufficiently small perturbations cannot change the signature of a non-null hypersurface.

\section{Implementation in Higher Dimensions}
\label{sec:generald}

Our analysis and interpretation in the previous section relied heavily on differential entropy. The definition of differential entropy and its relation to the area of bulk surfaces do not have known generalizations to more than three bulk dimensions (except in the special case of generalized planar symmetry~\cite{MyeHea14}). Nonetheless, as the coarse-graining procedure itself is well-defined in any dimension (and indeed, in any field theory) we can still make progress in higher dimensions.

Subregion/subregion duality and in particular the reconstructibility of the entanglement wedge $W_{E}[R]$ from the reduced density matrix $\rho_{R}$ immediately imply that any coarse-graining family $F$ in a holographic CFT (with a semiclassical bulk dual) gives rise to natural bulk geometric constructs associated to the information recovery limit of the coarse-grained state.  Such regions were alluded to at the end of Section~\ref{sec:coarsegraining} and we now treat them in more detail.  First, let us define two regions of interest:
\begin{defn}
Let~$F = \{R_\lambda\}$ be a coarse-graining family.  In any state $\rho$ with a semiclassical gravitational dual, we define the \textit{reconstruction region}~$L[F] \equiv \cup_\lambda W_E[R_\lambda]$ as well as the unreconstructible region\footnote{The region~$U[F]$ was recently considered in~\cite{NomRat18}; below we will comment on the connection of the area laws of~\cite{NomRat18} to ours.}~$U[F] \equiv \cap_\lambda W_E[\, ^s \overline{R}_\lambda]$.
\end{defn}
\noindent By entanglement wedge reconstruction, the reconstruction region~$L[F]$ consists precisely of all points in the bulk at which local bulk operators are known to be recoverable from the reduced states~$\rho_{\lambda}$ constituting $\rho_{F}$.  While the nomenclature may suggest otherwise, not \textit{all} operators in $L[F]$ are recoverable: sufficiently smeared operators will require access to reduced states over unions of the $R_{\lambda}$, which are not data accessible to $\rho_{F}$\footnote{As an example, consider a bilocal operator $\Ocal(x_{1},x_{2})$ evaluated any two points that do not live in a single $W_{E}[R_{\lambda}]$ for some $\lambda$.}.  Conversely, the unreconstructible region~$U[F]$ (which in pure states can be written as~$U[F] = \cap_\lambda \, ^s \overline{W_E[R_\lambda]}$) consists of all points in the bulk at which no local operators can be recovered from~$\rho_F$.

A natural expectation from the holographic nature of gravity is that there exists a special surface (of some codimension), defined geometrically from the boundaries of these regions, which characterizes the amount of information lost in the coarse-graining.  Intuition from the differential entropy in three bulk dimensions as well as from general gravitational thermodynamics suggests that such a surface should be codimension-two, and that its area is a measure of the data rendered inaccessible by our coarse-graining.

In general, how might such a surface be obtained?  Since~$U[F]$ is a domain of dependence, one option is the ``rim'' of~$U[F]$.  Another option, motivated by the construction of the curves~$\sigma_B$ in three bulk dimensions, is to construct a surface (in fact, multiple surfaces) as the locus of points where~$\partial L[F]$ changes signature from null to timelike, or alternatively as the locus of points where the generators of~$\partial J^\pm[L[F]]$ leave~$\partial L[F]$.  While the ``rim'' of~$U[F]$ is more natural, the latter option is clearly the correct one to consider in the case of the differential entropy.  Finally, yet another natural surface is provided by the intersection~$\partial J^+[L[F]] \cap \partial J^-[L[F]]$, which is the ``rim'' of the region~$^s \overline{F[L]}$ spacelike separated from~$L[F]$.  We sketch all these surfaces in Figure~\ref{subfig:generalunion}.
\begin{figure}[t]
\centering
\subfigure[]{
\includegraphics[page=18]{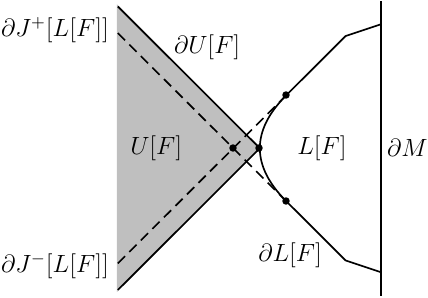}
\label{subfig:generalunion}
}
\hspace{1cm}
\subfigure[]{
\includegraphics[page=19]{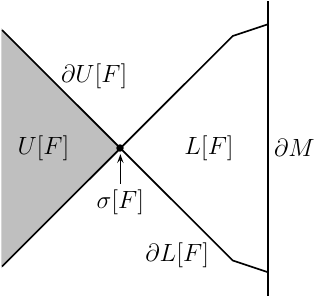}
\label{subfig:splitting}
}
\caption{In general,~$L[F]$ and~$U[F]$ look as shown in \subref{subfig:generalunion};~$U[F]$ is globally hyperbolic, but the boundary of~$L[F]$ may have timelike portions (this figure can be thought of as a ``constant-angle'' slice through Figure~\ref{fig:wedgeunionintersections}).  From these regions, we may define the assorted surfaces mentioned in the text, shown here as dots.  For an appropriate choice of~$\{R_\lambda\}$,~$U[F]$ and~$L[F]$ ``touch'' along a codimension-two surface~$\sigma$, as shown in~\subref{subfig:splitting}; in this case, all the surfaces discussed in the text coincide with~$\sigma$.}
\label{fig:splitting}
\end{figure}
With a plethora of options to choose from, it is not immediately clear which surface encodes the missing data in general dimension.  It is also certainly possible that more than one surface is relevant\footnote{Indeed, it's not too difficult to see that even in three bulk dimensions, in many cases the surface~$\partial J^+[L[F]] \cap \partial J^-[L[F]]$ is not irrelevant, as it is constructed by the intersection of null congruences fired from the differential entropy surfaces~$\sigma_B^\pm$.  These null congruences have negative expansion towards~$\partial J^+[L[F]] \cap \partial J^-[L[F]]$, and thus its area is bounded above by the differential entropy.}.  To develop a complete understanding, we would need a monotonic field-theoretic object constructed from our coarse-graining family $F$ in any dimension; in Section~\ref{sec:disc} we will comment more on possible approaches towards finding such a construct (these are clearly related to the $a$, $c$, and $F$-theorems).  For now, to make some progress, we will take a conservative approach: we focus on the case where all these candidates are coincident on a single surface~$\sigma$, defined as
%DRAGONS 
\begin{defn}
\label{def:edge}
If the intersection
\be
\sigma[F] \equiv \partial L[F] \cap \partial \, ^s \overline{L[F]}
\ee
is a nonempty, compact, codimension-two spacelike surface, then we call $\sigma[F]$ the \textit{reconstruction edge} of $L[F]$ (the ``edge of $L[F]$'' for short).  In such a case, we also have~$\sigma[F] = \partial J^+[L[F]] \cap \partial J^-[L[F]] = \partial J^+[U[F]] \cap \partial J^-[U[F]] = \partial L[F] \cap \partial U[F]$\footnote{In this degenerate case, $L[F]$ is a subset of the bulk-to-boundary definition of the ``outer wedge'' of $\sigma[F]$, and $^{s}\overline{L[F]}$ coincides with the ``inner wedge'' of $\sigma[F]$~\cite{EngWal17b, EngWalTA}.}. %Otherwise, EVERYTHING IS EDGELESS AND SAD.
\end{defn}

\noindent Intuitively, an edge separates the degrees of freedom removed by our IR coarse-graining procedure from those removed by Wilsonian-like holographic RG; this is the reason that, as noted in Section \ref{sec:intro}, the area laws presented below and those of \cite{NomRat18} occasionally coincide (this does not apply to the more general area laws of Section~\ref{sec:SSA}).  However, due to the differences between the two constructions, it is not clear in general when exactly the two classes of area laws -- the ones presented in this section and those of \cite{NomRat18} -- agree.

Let us now prove our area laws, corresponding to a monotonicity along continuous IR coarse-grainings.  To do so, we first prove some preliminary properties about $\sigma[F]$ itself.
\begin{lem}
\label{lem:edge}
Let~$\sigma[F]$ be the edge of~$L[F]$.  Then~(i)~$\sigma[F]$ is~$C^1$ on all but a sparse set, and~(ii)~every point on~$\sigma[F]$ lies on at least one of the~$X[R_\lambda]$, and where~$\sigma[F]$ is~$C^1$ it is tangent to any~$X[R_\lambda]$ that intersects it.
\end{lem}

\begin{proof}
Item~(i) follows immediately from Proposition~6.3.1 of~\cite{HawEll} and the fact that~$\sigma[F]$ is the intersection of a past set and a future set (see~\cite{Per04} for the desired result on sparseness).  Next, by definition~$\sigma[F] \subset \cup_\lambda W_E[R_\lambda]$, and therefore for every~$p \in \sigma[F]$, there is an~$R_\lambda$ such that~$p \in W_E[R_\lambda]$.  Since~$p \in \sigma[F] = \partial J^+[L[F]] \cap \partial J^-[L[F]]$,~$p$ cannot be timelike-related to any point in~$L[F]$.  But the only points in~$W_E[R_\lambda]$ which are not timelike related to any others are those on the HRT surface~$X[R_\lambda]$, and thus~$p \in X[R_\lambda]$ for some~$R_\lambda$. Finally,  we want to show that if $\sigma[F]$ is $C^{1}$ at a point $p$, then it is tangent to any $X[R_{\lambda}]$ that intersects it. Since by definition none of the $X[R_{\lambda}]$ intersect the interior of $^{s} \overline{L[F]}$, they cannot cross $\sigma[F]$.  Since the $X[R_{\lambda}]$ are themselves $C^{1}$ everywhere, any $X[R_{\lambda}]$ that intersects $\sigma[F]$ at a point $p$ where $\sigma[F]$ is $C^{1}$ must be tangent to $\sigma[F]$ at $p$.
\end{proof}

\begin{figure}[t]
\centering
\includegraphics[page=20]{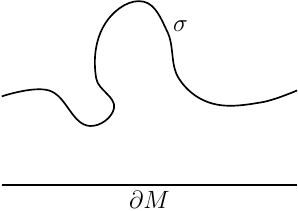}
\caption{Edge surfaces are prohibited from sampling different ``bulk depths'' as shown here; more precisely, they must be normal surfaces.}
\label{fig:RGsplitting}
\end{figure}

We will now see that the defining feature of an edge of a coarse-graining family  prevents it from ``bending'' too much in the spacetime, as shown in Figure~\ref{fig:RGsplitting}.  More precisely, this means that edges are so-called normal surfaces:
\begin{lem}
\label{lem:positiveexpansion}
Let~$\sigma[F]$ be the edge of $L[F]$.  Then the null expansion off of $\sigma[F]$ is non-negative in both directions towards $L[F]$, and in a sufficiently small tubular neighborhood of~$\sigma[F]$ each of the outgoing null congruences fired from~$\sigma[F]$ are~$C^1$ (except possibly on~$\sigma[F]$ itself).
\end{lem}
\begin{proof}
$\partial L[F]$ is locally (near $\sigma[F]$) a union of two null congruences. Let $N$ be one of these two null congruences in some small neighborhood of $\sigma[F]$, and take the generators of $N$ to be future-directed away from $\sigma[F]$ (the proof is identical for the past-directed generators).  First we show by contradiction that the expansion of~$N$ is non-negative wherever~$\sigma[F]$ is~$C^1$.  Assume that there exists a point~$p \in \sigma[F]$ at which~$\sigma[F]$ is~$C^1$ and $N$ has negative expansion away from $\sigma[F]$.  By Lemma~\ref{lem:edge} there exists an HRT surface~$X[R_\lambda]$ tangent to~$N$ at~$p$. But by Theorem~1 of~\cite{Wal10QST},~$X[R_\lambda]$ must intersect the future of~$N$ (and thus of~$\sigma[F]$), implying that~$X[R_\lambda]$ leaves~$L[F]$.  We immediately obtain a contradiction with the definition of $L[F]$: thus~$N$ has non-negative expansion wherever~$\sigma[F]$ is~$C^1$.

Now we show directly that~$N$ has positive expansion whenever~$p$ lies on a cusp.  If~$p$ lies on a cusp, by Lemma~\ref{lem:edge} there exists an HRT surface~$X[R_\lambda] \in L[F]$ for some~$R_\lambda \in F$ such that~$p \in X[R_\lambda]$.  Since~$\sigma[F]$ and~$X[R_\lambda]$ are achronally separated, there exists a Cauchy slice~$\Sigma$ containing them both.  Now consider a neighborhood~$U_p$ of~$p$ sufficiently small so that the geometry of~$\Sigma \cap U_p$ is approximately flat.  Since the set of points on which~$\sigma[F]$ is not~$C^1$ is sparse,~$U_p$ contains points at which~$\sigma[F]$ is~$C^1$, and therefore at which~$\sigma[F]$ has a well-defined normal (directed towards~$L[F]$) in~$\Sigma$.  As shown in Figure~\ref{fig:lemmaproof}, the fact that~$U_p$ is small allows us to discuss the convergence or divergence of these normals near~$p$.  Now, in order for~$\sigma[F]$ to lie outside of~$\Int[L[F]]$, it must lie to one side of~$X[R]$; this means that the normals to~$\sigma[F]$ on~$\Sigma$ must be diverging.  Thus the null congruences fired off the~$C^1$ portions of~$\sigma[F] \cap U_p$ do not intersect each other, implying that the null congruence~$N$ must also consist of generators originating at~$p$.  Since these generators originate at a caustic, they must have positive expansion.

Finally, it is clear that where~$\sigma[F]$ is~$C^1$,~$N$ must locally be $C^{1}$ as well.  Where~$\sigma[F]$ has a cusp, the above implies that the generators of~$N$ are diverging, so~$N$ is locally~$C^1$ as well.
\end{proof}

\begin{figure}[t]
\centering		
\includegraphics[page=21]{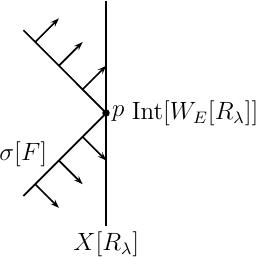}
\caption{In a sufficiently small neighborhood of a point~$p$ lying on a cusp of~$\sigma[F]$, the normals to~$\sigma[F]$ must be diverging in order for~$\sigma[F]$ to lie outside the HRT surface~$X[R_\lambda]$.}
\label{fig:lemmaproof}
\end{figure}

We now obtain the final result:
\begin{thm}
\label{thm:area}
Let $F(r)$ be a continuous IR coarse graining where each of the corresponding $\sigma(r)$ is an edge of $L[F(r)]$.  Then
\be
\frac{d \Area[\sigma(r)]}{dr} \ge 0.
\ee
\end{thm}

\begin{proof}
Let~$H$ be the hypersurface foliated by~$\{\sigma(r)\}$, and let~$h^a$ be the outward-directed normal vector field in~$H$ to the~$\{\sigma(r)\}$.  Note that~$h^a$ cannot be timelike anywhere, since under a coarse-graining the edges~$\sigma(r)$ must move outwards (i.e.~for any~$r_2 > r_1$,~$\sigma(r_2) \subset L[F(r_1)]$, so~$\sigma_1$ and~$\sigma_2$ cannot be timelike-separated).  Since the~$\sigma(r)$ need not be everywhere~$C^1$,~$h^a$ may be singular (on a set of measure zero); however, Lemma~\ref{lem:positiveexpansion} guarantees that outward-directed integral curves of~$h^a$ only start and never end at such singularities of~$h^a$.

Now, by Lemma~\ref{lem:positiveexpansion}, the expansions of the outwards-directed null congruences from each of the~$\sigma(r)$ are non-negative.  But the expansion of the integral curves of~$h^a$ is just a linear combination of these two null expansions (with non-negative coefficients).  Therefore the expansion along the~$h^a$ congruence is non-negative as well, and so the area of~$C^1$ portions of the~$\sigma(r)$ is non-decreasing.  Moreover, since generators of the~$h^a$ congruence never leave~$H$ as they flow outward, singularities of~$h^a$ where new generators appear can only increase the area of the~$\sigma(r)$.  Thus the areas of the~$\sigma(r)$ are non-decreasing under a flow along~$h^a$.
\end{proof}

This result is in fact an infinite family of area laws.  These area laws apply to non-timelike foliations, but in particular they can include causal horizons: in certain cases, the Hawking area law for causal horizons is a special case of these general area laws (e.g. the early stages of AdS-Vaidya collapse).  As an aside, note that it is simple to show that the so-called outer entropy of slices of causal horizons is bounded from above by their area~\cite{EngWalTA}, so that our area law immediately also suggests an outer entropy increase theorem as well; indeed, such a result was found in~\cite{NomRem18}.

In this section, we had to resort to bulk arguments to prove an area law, in contrast with our results in three dimensions, in which the area law simply manifested from SSA in the boundary theory. Moreover, while in the three dimensional case we could understand SSA as the dual of the area law, we have no such interpretation here. This is due to the absence of an appropriate generalization of the differential entropy to higher dimensions. The area monotonicity property, however, coupled with the obvious significance of these degenerate surfaces to our coarse-graining procedure, suggests that there exists a higher dimensional analogue of the differential entropy. In fact, our coarse-graining procedure and the area monotonicity theorem are sufficiently constraining that we expect to be able to use them to \textit{find} the requisite quantity.

\section{Quantum Generalization}
\label{sec:quantum}

We have so far neglected the backreaction of bulk quantum fields on the geometry.  While this limit is instructive, any results derived in it should be robust under quantum corrections to the geometry in order to be physically significant.  In this regime, a fluctuation to the spacetime metric is viewed as an operator whose expectation value is well-approximated by an expansion in powers of $G_{N}\hbar$.

The holographic computation of the von Neumann entropy therefore incorporates quantum corrections to the HRT surface; consequently the Dong-Harlow-Wall argument for entanglement wedge reconstruction includes quantum corrections as well. Since the edge surface $\sigma[F]$ is defined via HRT surfaces, it too changes under quantum corrections. This is consistent with the rule of thumb that perturbative quantum gravity effects can violate area monotonicity theorems: the quantity which is monotonic is a ``quantum-corrected area'' known as the generalized entropy of a surface $\sigma$\footnote{Note here that $\sigma$ must be Cauchy-splitting for this to be well-defined: i.e. $\sigma$ must divide a Cauchy slice into two distinct components.}, defined as
\be
S_{\mathrm{gen}}[\sigma] \equiv \frac{\mathrm{Area}[\sigma]}{4G_{N}\hbar} + S_{\mathrm{out}}(\sigma),
\ee
where $S_{\mathrm{out}}(\sigma)$ is the von Neumann entropy of the propagating quantum fields on a Cauchy slice of the exterior of $\sigma$.  $S_{\mathrm{gen}}$ was first defined by Bekenstein~\cite{Bek72, Bek73, Bek74}, and the approach of using it as a ``quantum-corrected area'' has since then been applied with remarkable success to generalize classical theorems to the semiclassical regime (see e.g.~\cite{Bek74, Wal10QST, EngWal14, BouFis15, BouEng15c}). While a comprehensive justification for the replacement of $A\rightarrow 4 G_{N}\hbar S_{\mathrm{gen}}$ is beyond the scope of this paper (see e.g.~\cite{BouFis15} for a recap), we cannot resist pointing out that there is significant evidence that the combined quantity $S_{\mathrm{gen}}$ is UV-finite; this provides evidence that the appropriate generalization of the area in perturbative quantum gravity is the generalized entropy.  Further evidence is provided by the appropriate quantum generalization of the HRT prescription: HRT surfaces (which classically extremize the area functional) are replaced by \textit{quantum} extremal surfaces, which extremize~$S_\mathrm{gen}$~\cite{EngWal14}, proven recently in~\cite{DonLew17} (the earlier work of~\cite{FauLew13} computed the corrections at first order, where the generalized entropy of the classical and quantum extremal surfaces agrees).  The entropy of a boundary region~$R$ is then computed by~$S_\mathrm{gen}[X_\mathrm{quant}[R]]$, with~$X_\mathrm{quant}[R]$ the quantum extremal surface homologous to a Cauchy slice of~$R$~\footnote{The prescription of~\cite{EngWal14} instructs us to compute the von Neumann entropy from the generalized entropy of the quantum extremal surface with smallest $S_{\mathrm{gen}}$. Here we ignore the subtleties with quantum extremal surfaces whose generalized entropy differs by $\Ocal(1)$ bits.}.

We should thus expect that the natural generalization of the area theorem in Section~\ref{sec:generald} is a Generalized Second Law: a monotonicity theorem for $S_{\mathrm{gen}}$. This is indeed the case, as we shall now show.

First, let us note that because~$S_\mathrm{gen}$ is not locally defined, the appropriate quantum generalization of the classical expansion -- the so-called quantum expansion~\cite{Wal10QST, EngWal14, BouFis15} -- requires a functional derivative of~$S_\mathrm{gen}$ under local deformations of~$\sigma$.  That is, the quantum expansion of a surface~$\sigma$ at~$p \in \sigma$ in the null direction~$k^a$ is defined as
\be
\Theta_k[\sigma,p] \equiv \frac{4G_{N}\hbar}{\delta A} \frac{\delta S_{\mathrm{gen}}[\sigma]}{\delta \sigma_k(p)},
\ee
where~$\delta S_\mathrm{gen}/\delta \sigma_k(p)$ schematically refers to a deformation of~$\sigma$ at~$p$ in the~$k^a$ direction of area~$\delta A$; see~\cite{BouFis15} for the precise definition.  Bekenstein's famous GSL for causal horizons is equivalent to the statement that $\Theta$ is non-negative on future causal horizons (and nonpositive towards the future on past causal horizons).  Likewise, the quantum extremal surfaces~$X_\mathrm{quant}[R]$ mentioned above have vanishing quantum expansion in all null directions.

We may now examine the quantum generalizations of the results of Section~\ref{sec:generald} (we will save comments on quantum generalizations of Section~\ref{sec:SSA} for later).  First, note that Definition~\ref{def:edge} remains unchanged:~$\sigma[F]$ is still the edge of the union of entanglement wedges~$L[F]$, although now these entanglement wedges are obtained from \textit{quantum} rather than classical extremal surfaces.  Lemma~\ref{lem:edge} then remains unaltered under the mild assumption that quantum extremal surfaces are~$C^1$.\footnote{Technically, the location of a quantum extremal surface is ``fuzzy'' due to fluctuations in the spacetime metric; see~\cite{EngWal14, BouFis15, Lei17} for a discussion. Here we will treat these under the assumption, justified in greater detail in the above references, that a notion of tangency can still be defined.}  Similarly, Lemma~\ref{lem:positiveexpansion} remains unchanged, except that the the null expansion of~$\sigma$ is replaced by the quantum expansion.  The reason for this straightforward modification is that the crucial result, which compares the classical expansion of tangent null hypersurfaces, admits a quantum generalization in terms of their quantum expansion~\cite{Wal10QST}.  This establishes the result wherever~$\sigma$ is~$C^1$, while at cusps the classical expansion must be strictly positive, and thus perturbative quantum corrections will not alter its sign.  Finally, it then follows that Theorem~\ref{thm:area} is replaced by a GSL:
\begin{thm}
Let $F(r)$ be a continuous IR coarse graining where each of the corresponding $\sigma(r)$ is an edge of $L[F(r)]$. Then
\be
\frac{d S_\mathrm{gen}[\sigma(r)]}{dr} \ge 0.
\ee
\end{thm}
The proof is essentially the same as for the classical case, except that the non-negativity of the quantum expansion of~$\sigma(r)$ in both null directions towards~$L[F(r)]$ guarantees that the generalized entropy increases rather than the area.  The technology is essentially the same ``zigzag argument'' used in~\cite{BouEng15c}.

Thus we have found that the area law, Theorem~\ref{thm:area}, associated to our IR coarse-graining admits a quantum generalization as a GSL.  A natural question is whether the precise entropic connection to SSA via differential entropy exists in a three-dimensional perturbatively quantum bulk: after all, the monotonicity of~$S_\mathrm{diff}$ from SSA is a general, purely field theoretic statement that makes no requirement on the bulk (or even the existence of one).  However, the key dictionary entry used to translate the monotonicity of~$S_\mathrm{diff}$ to a bulk area law -- that is, the mapping of~$S_\mathrm{diff}$ to the area of the differential entropy surfaces -- must receive quantum corrections whose behavior is at this point unclear.  We plan to investigate these corrections in future work.

\section{Discussion}
\label{sec:disc}

Coarse-graining is expected to be the fundamental aspect of quantum gravity that permits the emergence of semiclassical spacetime, i.e. the regime in which the UV data of quantum gravity decouples from the low-energy degrees of freedom. A significant challenge in any attempt to understand this process is the lack of a precise notion of UV gravitational data. In this paper, we have used AdS/CFT to investigate this question via the boundary theory, circumventing this problematic issue. The governing principle under which we operate is that, by the UV/IR correspondence, this data is encoded in the CFT IR.  This motivates a precise IR coarse-graining, which through the lens of quantum error correction can be viewed as the erasure of bulk data.

Explicitly, by restricting our knowledge of a state $\rho$ to reduced density matrices of a set of regions (i.e.~a coarse-graining family), we remove IR data such as long-range correlation functions and long-range entanglement. The AdS/CFT dictionary entry of subregion/subregion duality automatically translates this into an erasure of a region of the bulk interior. Erasing larger bulk regions corresponds to discarding a larger sector of the boundary IR.

Regardless of how well-motivated our procedure may be, without evidence that it makes contact with the actual coarse-graining mechanism built into quantum gravity, it is nothing more than a framework for removing IR information in a quantum field theory. We find that this is not the case, as our coarse-graining procedure passes a highly nontrivial test: it gives rise to holographic area monotonicity theorems in the classical regime, and Generalized Second Laws in perturbative quantum gravity.

Does every area law have a statistical significance in quantum gravity, and if not, why should we expect that ours do? The coincidence of a well-motivated coarse-graining procedure and its realization as an area monotonicity property (which behaves correctly under quantum corrections) is too strong to ignore. However, we accept that some of our readers may remain skeptical at this stage. Not to fear; the connection goes deeper.

In three bulk dimensions, the increase of area of a family of surfaces corresponding to a particular (continuous) coarsening is \textit{precisely} a result of strong subadditivity of the von Neumann entropy. Since strong subadditivity is a measure of the irreversibility of the removal of a subsystem, our area laws in three bulk dimensions are exactly the gravitational statement of irreversibility of the coarse-graining. Whether or not these particular (three-dimensional) area laws were suspected of having statistical quantum gravity significance prior to our work, the conclusion is now inevitable: they are a result of statistical coarse-graining. Moreover, our mechanism is a generalized version of the one which gives rise to the Casini-Huerta version of the $c$-theorem~\cite{CasHue04,CasHue06}, again indicating the connection with the irreversibility of this IR coarse-graining.

This roughly summarizes the framework and its justification; we now briefly comment on some interesting applications beyond the existence of a dual area theorem.

\paragraph{Mixed Signature Area Law:} Our construction of mixed-signature area laws is particularly intriguing due to their relation to holographic screens~\cite{Bou99d}, which are essentially local analogs of event horizons and which exist in general spacetimes and can have mixed signature.  As shown in~\cite{BouEng15a,BouEng15b}, holographic screens obey an area law irrespective of their signature: the area of so-called future holographic screens is always increasing towards the past on timelike portions and outwards on spacelike ones, and the time-reverse is true of past holographic screens.  An entropic explanation of this law on \textit{spacelike} portions~\footnote{This is the area law for so-called future outwards trapping horizons~\cite{Hay93} or dynamical horizons~\cite{AshKri02}.} was given in~\cite{EngWal17b} in the context of AdS/CFT, but an explanation for the general mixed-signature case is lacking.  It is striking, however, that the directions of area growth (and by extension, of the signature changes) of past and future holographic screens are identical to those of the hypersurfaces~$H^\pm$ constructed from differential entropy.  We emphasize that while we only constructed explicit examples of such hypersurfaces in pure AdS, this behavior is general.  The universality of such mixed-signature area laws indicates that the same mechanism may give rise to them all.  We hope that this observation can be used to explain a mysterious aspect of holographic screens in general spacetimes.

\paragraph{Further Work on Differential Entropy in Two Dimensions:} 

In three bulk dimensions, differential entropy provided a remarkably crisp entropic interpretation of bulk area laws. However, while monotonicity of the differential entropy has a clear interpretation in terms of SSA, a precise physical interpretation of differential entropy itself has yet to be provided.  There are some hints that such an interpretation may exist.  For instance, in the special cases we have studied, the vanishing of~$S_\mathrm{diff}$ indicates that the coarse-graining family from which it is constructed is not a coarse-graining at all, i.e.~it contains essentially all IR data of the boundary theory.  Moreover, when evaluated in Poincar\'e invariant vacuum states,~$S_\mathrm{diff}$ is effectively an integrated version of the Casini-Huerta~$c$-function.

It is therefore desirable to develop an understanding of the differential entropy beyond the classical bulk regime. When the bulk is perturbatively quantum (and thus the von Neumann entropies of boundary regions are computed from quantum rather than classical extremal surfaces in the bulk), does~$S_\mathrm{diff}$ correspond to any interesting bulk object?  Such a bulk object would be monotonic by virtue of the fact that~$S_\mathrm{diff}$ is as well (recall that~$S_\mathrm{diff}$ is monotonic in any unitary relativistic QFT, regardless of the existence or regime of a dual bulk).  It is therefore natural to expect that the quantum-corrected bulk object dual to~$S_\mathrm{diff}$ should be a bulk generalized entropy~$S_\mathrm{gen}$, which would provide tantalizing evidence that~$S_\mathrm{diff}$ really is computing some fundamentally important object.  However, it is conceivable that more generally~$S_\mathrm{diff}$ may compute some other type of quantum-corrected area in the bulk, suggestive of other possible quantum generalizations of area beyond the generalized entropy; we leave an investigation of this question to future work.

\paragraph{Higher Dimensional Differential Entropy:}

The precise interpretation of area monotonicity in terms of SSA in three dimensions immediately calls for some extension to higher dimensions. To our knowledge, except for higher-dimensional configurations with sufficient symmetry to reduce to a three-dimensional problem~\cite{MyeHea14}, there is presently no such generalization.  To some extent, this is because so far~$S_\mathrm{diff}$ has been understood only in the holographic context as a dual computation of the area of certain bulk surfaces.  However, it is manifest from the results in this paper that~$S_\mathrm{diff}$ has crucially important monotonicity properties independent of the existence of any holographic dual.  This observation provides an invaluable guide in constructing higher-dimensional generalizations of differential entropy.  For instance, since in two boundary dimensions we may think of~$S_\mathrm{diff}$ as an integral over a generalized Casini-Huerta~$c$-function of the regions~$R_\lambda$, natural guesses for higher-dimensional objects would involve integrals over entropic~$F$-functions,~$a$-functions, etc.  We might hope that a judicious guess could then produce an object which still computes the area of bulk surfaces constructed from the entanglement wedges of the coarse-graining family used to define it, thus yielding an entropic interpretation of our area laws in higher dimensions.

\paragraph{Bulk Extent of Coarse-Graining Procedure:} In a previous paper~\cite{EngFis15}, we found that the differential entropy cannot compute areas of surfaces inside a bulk region associated to holographic screens. In particular, the differential entropy is insensitive to a subset of the interior of holographic screens. Since the presence of a single holographic screen implies the existence of an infinite set of them, there is some outer envelope whose interior cannot be probed by the differential entropy construction. For purposes of hole-ographic bulk reconstruction in dynamical gravity, this posed a serious problem. For our purposes, however, this is instead an interesting feature: our area laws must avoid this region of strong dynamical gravity.    This is possibly related to the non-locality of quantum gravity. Interestingly, it is possible in principle for our area laws to still approach close to a singularity by avoiding this hidden region; we plan to determine whether this is indeed the case in future work. 

%we were so stupid when we were young. Why oh why umbral region

\section*{Acknowledgements}

It is a pleasure to thank Ahmed Almheiri, Jos\'e Barb\'on, Raphael Bousso, Alex Belin, Aitor Lewkowycz, Esperanza Lopez, Don Marolf, Matthew Roberts, Eva Silverstein, Douglas Stanford, and Lenny Susskind for useful conversations and discussion.  NE thanks the MIT Center for Theoretical Physics, the Stanford Institute for the Theoretical Physics, and Imperial College London for hospitality during various stages of this work. The work of NE is supported in part by NSF grant PHY-1620059 and in part by the Simons Foundation, Grant 511167 (SSG).  SF was supported by STFC grant ST/L00044X/1 and thanks the University of California, Santa Barbara for hospitality while some of this work was completed.

\appendix

\section{Differential Entropy in AdS$_3$}
\label{app:mixedarealaw}

Here we provide details on the construction of the mixed-signature area law presented in~\ref{sec:SSA}.  First, note that general spacelike geodesics in global AdS~\eqref{eq:AdS3} can be obtained by boosting the geodesics~\eqref{eq:staticspacelikegeodesics} on constant-time slices; these general geodesics are given by~$(t,r_*,\phi) = (T(s),R_*(s),\Phi(s))$ with
\begin{subequations}
\label{eqs:generalgeodesics}
\begin{align}
\tan((T(s) - t_0)/\ell) &= s \tan (\Delta t/ 2\ell), \\
\tan (\Phi(s) - \phi_0) &= s \tan (\Delta \phi/2), \\
\sin(R_*(s)/\ell) &= \sqrt{\frac{\cos^2(\Delta \phi/2) + s^2 \sin^2 (\Delta \phi/2)}{\cos^2(\Delta t/2\ell) + s^2 \sin^2 (\Delta t/ 2\ell)}},
\end{align}
\end{subequations}
where~$s \in (-1,1)$ is a (non-affine) parameter along the geodesic and the endpoints of the geodesic lie at~$(T(\pm 1), \Phi(\pm 1)) = (t_0 \pm \Delta t/2, \phi_0 \pm \Delta \phi/2)$.  In order for these endpoints to be spacelike-separated, we require~$0 \leq \Delta t/\ell < \pi - |\Delta \phi - \pi|$.  Equivalently, in terms of the null separations~$\Delta u = \Delta t/\ell + \Delta \phi$,~$\Delta v = \Delta t/\ell - \Delta \phi$, we must have~$0 < \Delta u < 2\pi$,~$0 > \Delta v > -2\pi$.  As described in the main text, symmetry considerations allow us to restrict just to the regime~$\Delta \phi \leq \pi$, or~$\Delta u \leq 2\pi - |\Delta v|$.

Now, consider the family of intervals shown in Figure~\ref{fig:boostedintervals}.  Using~$\lambda$ as a parameter along this family, we have
\be
\label{eq:boostedintervals}
t_0(\lambda) = 0, \quad \phi_0(\lambda) = \lambda, \quad \Delta t(\lambda)/\ell = \frac{1}{2}(\Delta u + \Delta v), \quad \Delta \phi(\lambda) = \frac{1}{2}(\Delta u - \Delta v).
\ee
To obtain the curves~$\sigma_B^\pm$, we follow the construction outlined in Section~4 of~\cite{MyeHea14}.  The deviation vector~$(\partial_\lambda)^a$ along the family of bulk geodesics is simply given by~$(\partial_\phi)^a$, and therefore its component~$\eta^a$ normal each geodesic is given by
\be
\eta^a = (\partial_\phi)^a - \frac{\partial_\phi \cdot \xi}{\xi^2} \xi^a,
\ee
where~$\xi^a = (\partial_s)^a = T'(s) (\partial_t)^a + R_*'(s) (\partial_{r_*})^a + \Phi'(s) (\partial_\phi)^a$ is a tangent vector to the geodesics.  Computing~$\eta^2$, we 
find that on each geodesic~$\eta^2 = 0$ at
\be
s^\pm = \pm \left|\frac{\tan(\Delta t/(2\ell))}{\tan(\Delta \phi/2)}\right| = \pm \left|\frac{\tan((\Delta u + \Delta v)/4)}{\tan((\Delta u - \Delta v)/4)}\right|.
\ee
The locations of the bulk curves~$\sigma_B^\pm$ are found by inserting this expression back into~\eqref{eqs:generalgeodesics}, which yields the results~\eqref{subeq:timelikearealawt} and~\eqref{subeq:timelikearealawr} in the main text and reproduced here for convenience:
\bea
t_{\sigma_B} &= \pm \ell \arctan\left|\cot\left(\frac{\Delta u - \Delta v}{4}\right)\tan^2\left(\frac{\Delta u + \Delta v}{4}\right)\right|, \\
r_{\sigma_B} &= \frac{\ell}{2}\left|\cot\left(\frac{\Delta u}{2}\right) - \cot\left(\frac{\Delta v}{2}\right)\right|.
\eea

To compute the signature of the surfaces~$H^\pm$ traced out by~$\sigma_B^\pm$ with fixed~$\Delta v$ and varying~$\Delta u$, note that an orthogonal basis of tangent vectors on the~$H^\pm$ is provided by~$(\partial_\phi)^a$ and~$\chi^a = t_{\sigma_B}'(\Delta u) (\partial_t)^a + r_{\sigma_B}'(\Delta u) (\partial_r)^a$, where here we are thinking of~$t_{\sigma_B}(\Delta u)$ and~$r_{\sigma_B}(\Delta u)$ as functions of~$\Delta u$.  Since~$(\partial_\phi)^a$ is always spacelike, the signature of~$H^\pm$ is determined by the sign of~$\chi^2$, which comes out to
\be
\sgn(\chi^2) = \sgn\left[ 6\sin^2\left(\frac{\Delta v}{2}\right) + \sin^2\left(\frac{\Delta u + \Delta v}{2}\right) - 2\sin^2\left(\frac{\Delta u}{2}\right) \right].
\ee
Clearly~$\chi^2 > 0$ for~$\Delta u$ sufficiently close to~$0$ or~$2\pi-|\Delta v|$, so~$H^\pm$ are always spacelike near the asymptotic boundary and near~$r = 0$.  In order to change signature,~$\chi^2$ must therefore have a root in~$\Delta u$.  By turning the equation~$\chi^2 = 0$ into a quadratic in~$\sin^2 (\Delta u/2)$, it is straightforward to see that the sign of the discriminant is~$\sgn(2-9\sin^2(\Delta v/2))$, which is positive for~$\Delta v > -2\arcsin(\sqrt{2}/3) \approx -0.31 \pi$ or~$\Delta v < -2\pi + 2\arcsin(\sqrt{2}/3)$.  In the latter case, the values of~$\Delta u$ for which~$\chi^2 = 0$ lie outside the range~$(0,2\pi - |\Delta v|]$, and therefore~$H^\pm$ are everywhere spacelike.  On the other hand, in the former case~$\chi^2$ does indeed change sign for~$\Delta u \in (0,2\pi - |\Delta v|]$, so~$H^\pm$ have mixed signature.

\bibliographystyle{jhep}
\bibliography{all}

\end{document}